\documentclass{article}

\newcommand{\powdelta}[1]{\delta^{#1}}

    \usepackage[preprint]{neurips_2025}

\usepackage[utf8]{inputenc} 
\usepackage[T1]{fontenc}    
\usepackage{hyperref}       
\usepackage{url}            
\usepackage{booktabs}       
\usepackage{amsfonts}       
\usepackage{nicefrac}       
\usepackage{microtype}      
\usepackage{xcolor}         
\usepackage{amsmath}
\usepackage{amsthm}
\usepackage{algorithm}
\usepackage{algorithmic}
\usepackage{todonotes}
\usepackage{subcaption}
\usepackage{float}
\usepackage[capitalise,nameinlink]{cleveref}
\crefname{figure}{Figure}{Figures}
\crefname{equation}{Equation}{Equations}

\newtheorem{theorem}{Theorem}[section]
\newtheorem{corollary}[theorem]{Corollary}
\newtheorem{lemma}[theorem]{Lemma}
\newtheorem{claim}[theorem]{Claim}
\newtheorem{observation}[theorem]{Observation}
\newtheorem{question}[theorem]{Question}
\newcommand{\tee}{\bar{t}}

\title{Ads that Stick: Near-Optimal Ad Optimization through Psychological Behavior Models \thanks{A part of this work was done when Kailash was an intern at the Indian Institute of Science. Akash Pareek's research is supported in part by the KIAC Postdoctoral Fellowship, Ittiam Systems CSR grant, and the Walmart Center for Tech Excellence at IISc (CSR
Grant WMGT-23-0001). Arindam Khan’s research is supported in part by the Google India Research Award,
SERB Core Research Grant (CRG/2022/001176) on “Optimization under Intractability and Uncertainty”, Ittiam Systems CSR grant, and the Walmart Center for Tech Excellence at IISc (CSR
Grant WMGT-23-0001). }}

\author{%
  Kailash Gopal Darmasubramanian \\
  Department of Computer Science\\
  Indian Institute of Technology, Madras\\
  \texttt{cs22b098@smail.iitm.ac.in} \\
   \And
  Akash Pareek \\
  Department of Computer Science $\&$ Automation \\
  Indian Institute of Science, Bangalore\\
  \texttt{akashpareek@iisc.ac.in} \\
   \AND
  Arindam Khan \\
  Department of Computer Science $\&$ Automation \\
  Indian Institute of Science, Bangalore\\
  \texttt{arindamkhan@iisc.ac.in} \\
  \And
  Arpit Agarwal \\
  Department of Computer Science $\&$ Engineering \\
  Indian Institute of Technology, Bombay \\
  \texttt{agarpit@outlook.com} \\
}

\begin{document}

\maketitle

\begin{abstract}

Optimizing the timing and frequency of advertisements (ads) is a central problem in digital advertising, with significant economic consequences.
Existing scheduling policies rely on simple heuristics, such as uniform spacing and frequency caps, that overlook long-term user interest.
However, it is well-known that users' long-term interest and engagement result from the interplay of several psychological effects (Curmei, Haupt, Recht, and Hadfield-Menell,
ACM CRS, 2022).

In this work, we model change in user interest upon showing ads based on three key psychological principles: \emph{mere exposure}, \emph{hedonic adaptation}, and \emph{operant conditioning}.
The first two effects are modeled using a concave function of user interest with repeated exposure, while the third effect is modeled using a \emph{temporal decay} function, which explains the decline in user interest due to overexposure.
Under our psychological behavior model, we ask the following question: Given a continuous time interval $T$, how many ads should be shown, and at what times, to maximize the user interest towards the ads?

Towards answering this question, we first show that, if the number of displayed ads is fixed, then the optimal ad-schedule only depends on the operant conditioning function. 
Our main result is a quasi-linear time algorithm that outputs a \emph{near-optimal} ad-schedule,
i.e., the difference in the performance of our schedule and the optimal schedule is exponentially small.
Our algorithm leads to significant insights about optimal ad placement and shows that simple heuristics such as uniform spacing are sub-optimal under many natural settings. The optimal number of ads to display, which also depends on the mere exposure and hedonistic adaptation functions, can be found through a simple linear search given the above algorithm. We further support our findings with experimental results, demonstrating that our strategy outperforms various baselines. 

\end{abstract}
\section{Introduction}\label{sec:intro}

Digital advertising is the backbone of the trillion-dollar modern internet economy. It is often the primary channel used by advertisers for acquiring new customers and maintaining interest among existing users.
Given that user attention is a scarce commodity and the economic consequences of capturing it are significant, \emph{optimizing the timing and frequency of 
ads} is an important problem for advertisers. Furthermore, this is quite relevant across different contexts, such as:
\begin{itemize}
\item Inserting ads into a (live) video/ audio stream.
\item Sending push notifications to app users. 
\item Embedding sponsored content within a single user session or across user sessions.
\end{itemize}
In use-cases like the ones mentioned above, the goal is to maximize user interest/recall while accounting for potential \emph{fatigue/satiation} in the long run.

A long line of empirical work in behavioral psychology has shown that the temporal spacing and frequency of ads have a significant effect on the memory retention and fatigue of the customer 
\cite{singh1994enhancing, sahni2015effect, curmei2022towards}.
For instance, following an initial positive neural response to repeated stimuli, individuals tend to revert toward a baseline level of interest, causing the “thrill” of the same message to fade. This initial boost is known as \emph{mere exposure}, while the subsequent tapering-off is termed \emph{hedonic adaptation} in the behavioral psychology literature \cite{curmei2022towards}.
Moreover, insufficient spacing between two ads can drain attention and affect memory retention in the long run \cite{singh1994enhancing}. This effect is referred to as \emph{operant conditioning}.

Despite the importance of this problem, there has been relatively little work studying the long-term cognitive and psychological effects of ad impressions \cite{sahni2015effect, aravindakshan2015understanding, rafieian2023optimizing} (see \cref{sec:relatedwork} for more details). Most existing approaches for ad placement rely on simple heuristics such as uniform spacing, front-loading, or frequency caps \cite{aravindakshan2015understanding, rafieian2023optimizing, despotakis2021first}, or on short-sighted policies that treat each ad impression as independent of the others \cite{TheocharousTG15}.

Some recent studies have tried to move beyond these heuristics. For instance, Freeman, Wei, and Yang \cite{freeman2022does} ran an experiment with 327 participants to test where ads should be placed in order to reduce negative reactions such as \emph{anger, irritation, and perceived intrusiveness}. Their main hypothesis was that \textit{"Mid-roll ads will elicit more anger and be seen as more intrusive than preroll ads"}, and their experiments confirmed this. They also tested several other hypotheses, all arriving at the same conclusion that placing ads at the beginning is generally more effective than inserting them in the middle. Similarly, Ritter and Cho \cite{ritter2009effects} conducted an experiment with 129 participants on audio podcasts. Their hypothesis was that \textit{"Advertising at the beginning of podcasts will generate less intrusiveness, less irritation, more favorable attitudes toward an ad, and less ad avoidance than advertising in the middle of podcasts"}. Their findings supported this claim. Other studies also suggest similar trends. For example, Goldstein, McAfee, and Suri \cite{goldstein2011effects} argue that a mix of ad placements, at the beginning, at the end, and a small number evenly spaced in the middle, can create a more positive overall experience for users. These experimental studies motivated us to ask the following questions.

\begin{question}\label{que:1}
    Can we design a "useful" theoretical model for ad scheduling that captures the users’ psychological behavior?
\end{question}

\vspace{-0.3em}
A “useful” model should be simple, robust to changes, and consistent with trends observed in the real world. Ideally, the mathematical formulation should also admit a closed-form solution that is easy to interpret. The guiding principles captured by such a model can then serve as a foundation for the design and development of new algorithms.

To answer \cref{que:1} in the affirmative, we study a dynamic model of user interest under repeated ad exposures. Let $n+1$ ads be shown at time $\tee = (t_0, t_1, \dots, t_n)$, where each $t_i \in [0, T]$ and $t_i \leq t_{i+1}$ for all $i \in \{0, 1, \dots, n\}$. Here, $t_i$ represents the time at which the $i^{\text{th}}$ ad is shown, and without loss of generality, we set $t_0 = 0$ and $t_n=T$. We refer to $\tee$ interchangeably as a \emph{strategy} or a \emph{schedule}. Given a time horizon $T$ and a strategy $\tee$, the reward obtained from showing the $i^{\text{th}}$ ad at time $t_i$ is denoted by $R(\tee, i)$.
\footnote{Throughout, we use the term "reward" to describe the advertiser’s optimization objective. Depending on the context, this could represent the probability of purchase, user satisfaction with the product, or overall engagement with the ads.} The total reward associated with strategy $\tee$ is then defined as $R(\tee) = \sum_{i=0}^{n} R(\tee, i)$. Under the model proposed in \cref{que:1}, we ask the following question.

\begin{question}\label{que:2}
    Can we design an optimal schedule that maximizes the total reward?
\end{question}

\vspace{-0.3em}
In response to \cref{que:2}, we give a near-optimal algorithm to compute an ad schedule efficiently under our model. To capture user psychology in our model, we study $R(\tee, i)$ as a combination of two functions. The first function captures the (positive) mere exposure and the hedonic adaptation effects and depends only on the number of ads shown previously. We denote this function by $B: \mathbb{Z}_{\ge 0} \rightarrow \mathbb{R}_+$, a concave function with respect to the number of ads shown till now, and represents the reward when the $i^{th}$ ad is shown. The concavity of the function $B$ is justified by the `diminishing returns' property implied by mere exposure and hedonic adaptation, and is common in recent literature that considers the dynamic effect of actions \cite{patil2022mitigating, blumnearly}.

The second function is a temporal exponential decay function with parameter $\delta\in [0, 1]$ which captures the (negative) operant conditioning effect so that reward at time $t$ decreases by $\delta^{t-t'}$ for any ad shown previously at time $t'$.
This temporal decay function is motivated by the classical theory of \cite{ebbinghaus1913contribution} on forgetting curves, which hypothesizes the exponential decline of memory retention over time.
This model is similar to the influential \emph{goodwill stock} model of \cite{nerlove1962optimal},
and \emph{exponential discounting} models used in control theory \cite{leqi2021rebounding}.

The parameter $\delta$ plays a key role in capturing different types of user psychological behavior, such as anger, irritation, intrusiveness, or interest. A low value of $\delta$ indicates that past ad impressions have little effect on the user and therefore additional ads do not strongly reduce engagement. In contrast, a high value of $\delta$ implies that ads have more long-term effects on the user~\cite{curmei2022towards}. This makes a user highly sensitive when an ad is shown, which can quickly lead to irritation, loss of interest, or a perception of intrusiveness. Thus, the parameter $\delta$ provides a compact way to encode the psychological effects observed in experimental studies such as \cite{freeman2022does, ritter2009effects, goldstein2011effects}.

Given these functions, the goal is to find the number of 
ads $n+1$ and the strategy $\tee = (t_0, \dots, t_n) \in [0,T]$,
such that the total reward $R(\tee)$ is maximized. Given this model of user reward, we first show that the problem has a special structure:
if the number of ads $n+1$ is fixed, then the optimal display timing only depends on the function capturing operant conditioning.

Our main algorithmic result is a quasi-linear time algorithm that, given a fixed number of ads $n+1$, outputs a \emph{near-optimal} ad-schedule. Let $t_i^*$ denote the optimal time to place the $i$-th ad. We show that each $t_i^*$ can be approximated by $t_i$ (produced by our algorithm) with exponentially small error, namely, $t_i^*-\frac{1}{2^n}\le t_i \le t_i^*+ \frac{1}{2^n}$. Our result is based on several key insights: (1) the operant conditioning function is strictly convex, leading to a single global optima, (2) the global optima has a special structure such that each $t_i$ can be recursively found using $t_1, \dots, t_{i-1}$, (3) $t_1$ can be found using binary search with \emph{bounded error} and, (4) error
propagation in the recursive computation can be controlled. Next, to determine the optimal number of ads to display, which depends on the mere exposure and hedonic adaptation functions, we perform a simple linear search, using the result discussed above. We further support our findings by studying how various strategies perform in our theoretical model, and further demonstrate via plots that our ad placement strategy outperforms various baselines for $\delta \in [0,1]$.

Recall that experimental studies \cite{freeman2022does, ritter2009effects, goldstein2011effects} have suggested that placing ads at the beginning is generally more effective than placing them in the middle. While these findings provide useful evidence regarding effective ad placement, they are limited in scope and cannot establish that such a strategy is universally optimal across all scenarios. This motivates the need for a theoretical framework that can explain and generalize these observations. In particular, our model offers deeper insights into ad placement strategies by explicitly accounting for all values of $\delta \in [0,1]$.  To systematically understand how $\delta$ influences the optimal policy and how ad placement strategies may vary across different values of $\delta$, we analyze the optimal ad policy in our model for all values of $\delta \in [0,1]$. Our analysis demonstrates that the structure of the optimal schedule changes as $\delta$ changes. Building on this, our near-optimal ad scheduling algorithm provides not only practical scheduling strategies but also valuable theoretical insights into how ad placement should adapt across the full spectrum of $\delta$. To be precise, we observe that for small $\delta$, the ads in the optimal strategy $\tee$ are placed almost evenly in $[T]$. As $\delta$ increases,  more and more ads start to concentrate at $0$ and $T$ such that the first $t_i > 0$ moves towards 0, and the last $t_j < T$ moves towards $T$. The remaining ads are evenly spaced between $t_i$ and $t_j$. All these insights show that simple heuristics such as uniform spacing or placing more ads at the beginning (as observed by the prior experimental studies) are sub-optimal under many values of $\delta$, highlighting the need to adapt schedules based on user behavior (i.e., $\delta$).

\subsection{Related Work}\label{sec:relatedwork}

\noindent
\textbf{Behavioral psychology.}
While the work on behavioral psychology is vast, 
in our work we follow \cite{curmei2022towards} and focus on three well-studied phenomena from psychology: mere exposure, operant conditioning, and hedonic adaptation. 
Several works in behavioral psychology has empirically shown
the effect of mere exposure and hedonic adaptation 
under various settings \cite{cox2002beyond, fang2007examination, chugani2015happily, yang2015sentimental}.
The most relevant to our work is the study of \cite{hekkert2013mere} who found that 
attractiveness of a product increased with the number of times it was shown to a user. Moreover, in a similar context \cite{nelson2008interrupted} also found evidence of hedonic adaptation as a function of the number of exposures.  
Similarly, operant conditioning has also been well-studied (See \cite{cooper2007applied} and references therein). 
While operant conditioning might have additional connotations in the psychology literature, we mainly use it to model the `annoyance' or `satiation' effect of repeated exposures \cite{sahni2015effect}. 
There has been some recent effort on psychology-aware recommendation systems \cite{curmei2022towards, jesse2021digital}, 
however, incorporating psychological effects in advertising has received less attention. 

\noindent
\textbf{Value of $\delta$ in real world.} There exists some work related to determining the exact curve to analyze user retention of content. In the work of Curmei et. al.~\cite{curmei2022towards}, the authors note that in real-world scenarios, the value of $\delta$ is 0.98. Other efforts as seen in  ~\cite{murre2015replication}, ~\cite{goldstein2011effects} gives the value of delta between 0.7 and 0.99. It is also to be noted that it is in this regime where there is the greatest difference between various strategies according to our experiments. We finally note that in order to practically use our model, the experiments in ~\cite{murre2015replication} and ~\cite{goldstein2011effects} can be adapted.

\noindent
\textbf{Ad scheduling.}
The problem of optimizing ad schedules has been studied across various communities such as marketing, operations research and machine learning.
One of the earliest work is due to \cite{nerlove1962optimal},
who proposed the goodwill‐stock model that treats advertising as an investment that builds a “stock” of consumer goodwill (or awareness) which then depreciates over time. 
Specifically, they considered a dynamical model of change in consumer goodwill given an ad impression, and cast the problem of optimizing the ad schedule as an \emph{optimal control} problem. However, they do not consider memory/satiation effects, and the optimal policy under their model is to show most of the ads at the start of the time-horizon.
\cite{naik1998planning} extended this work to consider memory effects, however, the optimal policy under their dynamical system is not interpretable. 
\cite{sahni2015effect} conducted field experiments to demonstrate that \emph{temporal spacing} of ads has a large effect on the memory retention and satiation of the users.
More recently, the problem of optimizing ad schedules has been 
casted as a reinforcement learning problem \cite{rafieian2023optimizing, TheocharousTG15}.
While these algorithms tend to be general-purpose, their solutions are hard to interpret and can be difficult to execute in practice given the complexity of ad exchanges \cite{despotakis2021first}.


\noindent
\textbf{Dynamic rewards and restless bandits.}
Leqi et al. \cite{leqi2021rebounding} introduced a bandit framework that models user satiation using linear dynamical systems and showed that the greedy strategy is optimal when all arms have the same base reward and decay profile. More broadly, several works have explored reward structures where the current payoff depends on historical actions and decays over time \cite{heidari2016tight, levine2017rotting, seznec2019rotting}. In contrast, Kleinberg and Immorlica \cite{kleinberg2018recharging} study a setting where rewards increase with time since the last pull. Related ideas appear in models where rewards evolve based on the number of pulls or the time elapsed since the last interaction \cite{cella2020stochastic, basu2019blocking, warlop2018fighting, mintz2020nonstationary}.

\noindent
\textbf{Digital advertising.}
In \cite{schwartz2017customer}, the authors investigated customer acquisition through advertisements on online platforms using multi-armed bandits, and designed a policy that achieves an 8$\%$ improvement in acquisition rate.  \cite{adany2013allocation} addressed the problem of allocating personalized ads to users, considering each user's profile and estimated viewing capacity.  \cite{seshadri2015scheduling} explored advertisement scheduling to meet advertiser's campaign goals while maximizing ad-sales revenue.  \cite{dobrita2025federated} proposed a framework that leverages $k$-nearest neighbors to predict ad positions, enhancing ad scheduling optimization. Aiming to maximize viewership under budget constraints,  \cite{czerniachowska2019scheduling} presented a scheduling solution aligned with advertisers' budgets.  \cite{sumita2017online} developed a $(1 - 1/\epsilon)$-competitive algorithm for envy-free allocation of video ads where $\epsilon > 0$ is a constant.

\subsection{Organization}

In \cref{sec:problem}, we formalize the problem, and in \cref{sec:algo}, we present a near-optimal algorithm for scheduling ads. Next, in \cref{sec:merge} and \cref{sec:assume}, we analyze this algorithm, followed by a discussion of the broader implications of our work in \cref{sec:implication}. We support our theoretical findings with experiments in \cref{sec:experiment}, and conclude by outlining the key takeaways, limitations, and possible extensions in \cref{sec:conclusion}.

\section{Problem Description}\label{sec:problem}
Let $T$ denote the total time horizon, which is known to us in advance. We consider the setting where we have to display $n+1$ homogeneous ads in the continuous time interval $[0, T]$. We assume for now that displaying these ads is instantaneous, though in ~\cref{sec:implication}, we will show how to extend this to the case when each ad needs the same amount of time. 
As discussed before, the function $R(\tee, i)$ is composed of two parts. The first part $B(i)$ is the reward when the $i^{th}$ ad is shown (note that $i$ could also be $0$ and corresponds to the ad shown at $t_0$), and is simply a function of the number of times we have shown the ad previously, and not of the times at which we show ads. 

The second function depends on both the number of times we show ads and the times $t_i$ at which ads were shown. Let $\tee = (t_0,t_1,\dots,t_{n-1},t_n)$ denote the time and order in which ads were displayed. 

The second function is $\gamma.\sum_{i=0}^{l-1} \delta^{t_l-t_i}$, where $\delta \in [0,1]$ and $\gamma \in \mathbb{R}_{+}$ is a constant that parameterizes the strength of this effect. 
The term $\sum_{i=0}^{l-1} \delta^{t_l-t_i}$, $\delta \in [0, 1]$ captures the temporal decay in the loss when the $l^{th}$ ad is shown. Hence, the reward for the $i^{\text{th}}$ ad under strategy $\tee$ is given by
$R(\tee, i) = B(i) - \gamma \sum_{j = 0}^{i - 1} \delta^{t_i - t_j}.$

Let $\tilde{n}\in \mathbb{N}$ be an upper bound on the number of ads that can be shown in $[0, T]$. 
Our objective is to find the number of ads $n+1 \le \tilde{n}$, and the strategy $\tee = (t_0, \dots, t_n)$ according to which ads should be shown so that the reward $R(\tee)$ is maximized.

The reward $R(\tee)$ can be written as:
\begin{align*}
    R(\tee) &= \sum_{i=0}^n R(\tee,i)= \sum_{i=0}^n \left(B(i) - \gamma\sum_{j=0}^{i-1}\delta^{t_i - t_j}\right)= \left(\sum_{i=0}^nB(i)\right)- \left(\gamma\sum_{j<i}(\delta^{t_i-t_j}) \right).
\end{align*}
Observe that the first term $\sum_iB(i)$ is independent of the times at which we show the ads, and only depends on the number of ads themselves. The coefficient of the second term, i.e., $\gamma$, only plays the role of a scaling factor. Hence, maximizing $R(\tee)$ for a given value of $n$ is equivalent to minimizing the {\em loss} $L(\tee)$, where $L(\tee)$ is defined as  $L(\tee) = \sum_{j<i}(\delta^{t_i-t_j})$.

Once we can find the optimal value of $R(\tee)$ for strategies that involve showing $n+1$ ads, we can iterate over $\tilde{n}$ possible values to find the optimal number of ads to be shown. Next, we provide an overview of the algorithm to find the optimal number of ads to be shown to maximize the reward, and the optimal strategy $\tee$ to show these ads.
\section{Near-Optimal Algorithm for Ad Scheduling}\label{sec:algo}

In this section, we provide an overview of our near-optimal algorithm for ad scheduling. We first consider the case when we know $n+1$, the number of ads to show. Given the time horizon $T$ and $\delta$, we want to compute the schedule $\tee = (t_0,t_1,\dots,t_n)$, such that for all $0\le i \le n-1$, $t_i \leq t_{i+1}$. 

The first step in our algorithm is to compute the number of ads to show at time $0$ and $T$. Note that multiple ads 
can be shown at the same time, given the instantaneous nature of the ads (we relax this assumption in  \Cref{sec:implication}).
Let $t_a>0$ be the first non-zero time at which an ad is shown, i.e., $a-1$ ads are shown at time $0$. We define the quantity $T_i= \delta^{t_i}$, and throughout the paper we derive our results in terms of $T_i$ for $i \in \{0, \dots, n\}$. 
Given the parameters $n$, $T$, $\delta$, our first algorithm (\Cref{alg:aT_a getting}) outputs the value of $a$ and $T_a=\delta^{t_a}$. Later in our analysis, we will show that all the subsequent ad timings can be found once we know $a$ and $T_a$. We now describe the algorithm to obtain $a$ and $T_a$.

\begin{algorithm}[H]
    \caption{Algorithm to obtain $a$ and $T_a$}
    \label{alg:aT_a getting}
    \begin{algorithmic}[1]
    \STATE \textbf{Input:} $n$, $\delta$, $T$.
   
    \STATE Find the smallest value of $a\in [n/2]$ s.t.\ $\frac{a^{n-2a}}{(a+1)^{n-2a}} > \delta^T$ and $\frac{1}{a^2}\cdot\frac{(a\delta^T)^{n+2-2a}}{(a\delta^T+1)^{n-2a}}< \delta^T$ holds. 
    \STATE For the above value of $a$, define the function $h(T_a) = \frac{1}{a^2}\cdot \frac{(aT_a)^{n-2a+2}}{(1+aT_a)^{n-2a}}$
    \STATE Compute the solution for the equation $h(T_a) = \delta^T$ via binary search. For binary search, initialize the search space to be $[\delta^{T}, 1]$. (see ~\cref{sec:approx} for more details)
    \RETURN $(a, T_a)$
    \end{algorithmic}
\end{algorithm}

\vspace{-1mm}
 As ~\cref{alg:aT_a getting} returns $a$, we now know the first non-zero time to show an ad.  Once $a$ and $T_a$ are known,  \cref{alg: optimal schedule} shows how to find the near-optimal schedule. Note that for the corner case when the value of $a$ does not exist, and for a detailed algorithm, refer to  \cref{sec:assume}.

\begin{algorithm}[H]
    \caption{Algorithm to obtain near-optimal schedule}
    \label{alg: optimal schedule}
    \begin{algorithmic}[1]
        \STATE \textbf{Input:} $n,\delta,T,a,T_a$
        \STATE Set $t_i =  0$, $\forall i\in \{0, \dots, a-1\}$ and $t_j = T$, $\forall j\in \{n-2a+1, \dots, n\}$
        \STATE Set $t_a = \ln (T_a)/ \ln (\delta)$ and $t_{n-a} = T-t_a$.
        \STATE Distribute the remaining $t_i, \forall i \in \{a+1,\dots,n-a-1\}$ between $t_a$ and $t_{n-a}$ such that they are equally spaced.
        \RETURN $\tee = <t_i>$
    \end{algorithmic}
\end{algorithm}

\cref{alg: optimal schedule} provides a near-optimal ad schedule when the number of ads $n$ is known in advance. If $n$ is unknown, we use  ~\cref{alg: optimal number of ads and schedule} to determine it. Below, we show how to compute the optimal number of ads $n$, given an upper bound $\tilde{n}$ on the total ads that can be shown within the time horizon $T$.

\begin{algorithm}[H]
    \caption{Optimal number of ads}
    \label{alg: optimal number of ads and schedule}
    \begin{algorithmic}[1]
        \STATE \textbf{Input:} $\delta,T,\tilde{n}$
        \STATE For each $n \in [\tilde{n}] $, compute the near-optimal strategy $\tilde{t_n}$ using ~\cref{alg: optimal schedule}
        \STATE For each near-optimal strategy, compute $R(\tilde{t_n})$
        \RETURN  the value of $n$ that maximizes $R(\tilde{t_n})$ 
    \end{algorithmic}
\end{algorithm}

\cref{alg: optimal number of ads and schedule} returns the optimal number of ads to show and the near-optimal strategy to display them.

For clarity, we focus our analysis on the case when $a = 1$, i.e., the case where only one ad is shown at time $0$ and time $T$, respectively. The general case, $a \neq 1$ is discussed in ~\cref{sec:assume} and follows a similar approach. When $a = 1$, our goal is to show that approximating $T_1$ helps us to approximate $t_1, \dots, t_{n}$, resulting in a nearly optimal schedule.  In the next section, we outline our approach for the $a = 1$ case, with full proofs provided in ~\cref{sec:merge}.

\section{Analysis of Our Algorithm for $a=1$}\label{sec:merge}
In this section, we present our analysis of the near-optimal strategy for the case $a = 1$. The analysis is divided into three parts: ~\cref{sec:uniqueminima} proves that $L$ has a unique minimum. ~\cref{sec:decoupling} derives closed-form expressions showing that each $T_i$ can be written in terms of $T_1$. ~\cref{sec:approx} first approximates $T_1$ and then $t_1$ to an additive error of $1/2^n$. Once $T_1$ and $t_1$ are approximated, the remaining $T_i$ and $t_i$ values can be computed easily. To this end, we show that $L$ has a unique minima.

\subsection{$L$ has a Unique Minima}\label{sec:uniqueminima}
In this section, we want to argue about the number of minima for the loss function $L$. Specifically, we want to show that $L$ is a strictly convex function, and hence, it has at most one minima. This would imply that local minima is also the global minimum \cite{boyd2004convex}. This characterization helps us to find the minima for $L$.

For a strategy ${\tee}$, to form a feasible solution, it must satisfy $t_i \leq t_{i+1}$ and $0 \leq t_i \leq T$. Then, the set of feasible solutions, which we denote by $\mathcal{D}$, is defined as follows:
$$t_i \in \mathbb{R}, 0 \leq t_i \leq T, t_i \leq t_{i+1}, t_0 = 0, t_n= T.$$

We first show why $t_0=0$ and $t_n=T$ are required to minimize $L$.

\begin{lemma}\label{lem:t0=0tn=n}
    In any optimal solution $\tee$ that minimizes the loss function $L$, we have $t_0=0$ and $t_n=T$.
\end{lemma}
\begin{proof}
      We first argue that $t_0=0$, in any optimal solution $\tee$. The argument for $t_n = T$ is quite similar.
    
    For the sake of contradiction, assume that we have an optimal solution $\tee$ where $t_0\neq 0$. Consider a solution $t'$ where $t'_i = t_i \text{ for all } i > 0$ and $t'_0 = 0$. 

   We have
    \[
    L(\tee) = \sum_{j>0} \delta^{t_j-t_0} + \sum_{j> i > 0}\delta^{t_j -t_i}
    \]
    and 
    \begin{align*}
    L(t') &= \sum_{j>0} \delta^{t'_j-t'_0} + \sum_{j> i > 0}\delta^{t'_j -t'_i}\\
    &= \sum_j \delta^{t_j} + \sum_{j> i > 0}\delta^{t_j -t_i}    \\
    &< \frac{1}{\delta^{t_0}}\sum_j \delta^{t_j}+ \sum_{j> i > 0}\delta^{t_j -t_i}    \\
    &= \sum_j \delta^{t_j-t_0} + \sum_{j> i > 0}\delta^{t_j -t_i}\\
    &= L(\tee).
    \end{align*}

    This implies that $L(t') < L(t)$, which is a contradiction.
\end{proof}

For $L=\sum_{i>j}\delta^{t_i-t_j}$, let $a_{ij} = t_i - t_j$, where $i>j$, and denote $L'=\sum_{i>j} \delta^{a_{ij}}$. The feasible solution $\mathcal{D'}$ for $L'$ is given by:
$$a_{ij} \in \mathbb{R} \hspace{.2cm} \forall i > j, a_{ij}\ge 0, a_{n0} = T, a_{ij}+a_{jk} =  a_{ik}$$

Using \cref{lem:t0=0tn=n}, we only consider those $\tee$ where $t_0 = 0$ and $t_n = T$. Since $t_0$ and $t_n$ are fixed, all references to $\nabla L$ hereafter are with respect to $t_1, t_2, \dots, t_{n-1}$.

To show that $L$ has a unique global minima, we divide our analysis into four parts:
\begin{itemize}
    \item First, we establish that there exists exactly one global minima in $\mathcal{D'}$ for the function $L'$, by showing that $L'$ is strictly convex and the space $\mathcal{D'}$ is compact.
    \item Next, we show a bijection between $\mathcal{D}$ and $\mathcal{D'}$. 
    \item We then show that, according to the above bijection, $L'$ indeed models $L$.
    \item Finally, we show that the minima of $L'$ corresponds exactly to the minima of $L$, and vice versa. 
\end{itemize}

We now prove each of these parts in detail. To start we first show the following lemma.

\begin{lemma}
\label{lem:unique minima for L'}
    $L'$ admits a unique minima in $\mathcal{D}'$.
\end{lemma}
\begin{proof}
Observe that $\delta^{a_{ij}}$ is strictly convex for constant $\delta \in (0,1)$, as the double derivative of $\delta^{a_{ij}}$ is positive. Since the sum of strictly convex functions is also strictly convex, $L'$ is a strictly convex function. 

Observe that $a_{ij} +a_{ni} + a_{j0} = a_{n0} =T$. Hence, $a_{ij}\leq T$ for all $i, j$. This means that the feasible region of $a_{ij}$ is bounded, and since the region is an intersection of linear inequalities, the feasible region of $a_{ij}$ is a polytope.

Also, note that $\mathcal{D}'$ describes a polytope and hence $\mathcal{D}'$ is a convex region. As $L'$ is strictly convex and $\mathcal{D}'$ is a convex region, $L'$ has at most one minima in $\mathcal{D}'$~\cite{boyd2004convex}. Again, we know that polyhedra are compact sets. Therefore, the function $L'$ has a unique minimum in $\mathcal{D}'$~\cite{boyd2004convex}.    
\end{proof}

We now present a bijection $\mathcal{B}$ between the space $\mathcal{D}$ and $\mathcal{D}'$ and show that $\mathcal{B}$ preserves minima. Define the mapping $\mathcal{B}:\mathcal{D}\rightarrow\mathcal{D}'$ as follows:

\begin{itemize}
    \item Given ${t_i}$, we set $a_{i0}= t_i$.
    \item Define $a_{ij} = a_{i0}-a_{j0}$. 
\end{itemize}

We now show that $\mathcal{B}$ is a bijection:

\[\mathcal{B}({\tee})_{ij} = t_i -t_j ,\forall i > j.\]

\begin{lemma}
\label{lem: B is bijection}
    The mapping $\mathcal{B}$ is a bijection from $\mathcal{D}$ to $\mathcal{D}'$.
\end{lemma}
\begin{proof}
    Let $\mathcal{B}({\tee})= {a}$, for an element ${\tee}\in \mathcal{D}$. By definition, $a_{ij} = t_i - t_j \geq 0$ since $i > j$. Also, note that $a_{ij} +a_{jk} =t_i - t_j +t_j -t_k = t_i - t_k = a_{ik}$. Hence, $a_{ij}$ satisfies the constraints associated with $\mathcal{D}'$. Hence, $\mathcal{B}$ is feasible in the domain $\mathcal{D}'$. 
    
     Observe that if $t_i\neq t'_i$, then $a_{i0} \neq a'_{i0}$, hence it is an injection.
     
     Given $a_{ij}$, let define $\tee$ where $t_i = a_{i0}$. We show that $\mathcal{B}({\tee}) = {a}$. Let $\mathcal{B}(\tee) = a'$, and we will show $a' = a$. Firstly, by definition of $a'$ and $t$, \[a'_{i0} = t_i = a_{i0}.\]

     We also have
     \[a'_{ij} = a'_{i0} - a'_{j0} = t_i-t_j = a_{i0} - a{j0} = a_{ij}.\]
     
     Therefore $a = a'$ and hence for each $a$, there exists a $\tee$ such that $\mathcal{B}(\tee) = a$. Therefore, $\mathcal{B}$ is a surjection. Since $\mathcal{B}$ is both an injection and a surjection, it is a bijection.
\end{proof}
\begin{corollary}
    As $\mathcal{B}$ is a bijection, it is also invertible. It can be verified that,
\[
(\mathcal{B}^{-1}(a))_i = a_{i0}.
\]
\end{corollary}

To this end, we now show that $L'$ indeed models $L$ according to $\mathcal{B}$:

\begin{lemma}
\label{lem: L' models L}
    $L(t) = L'(\mathcal{B}(\tee)).$
\end{lemma}
\begin{proof}
    \begin{align*}
        L(t) &= \sum_{i >j}\delta^{t_i-t_j}\\
        &=\sum_{i > j}\delta^{\mathcal{B}(t_i)-\mathcal{B}(t_j)}\\
        &= \sum_{i > j}\delta^{a_i0-a_j0}\\
        &=\sum_{i > j}\delta^{a_{ij}}\\
        &=L'(\mathcal{B}(\tee)).
    \end{align*}
\end{proof}

 Next, we demonstrate that $\mathcal{B}$ preserves the minima, i.e., $t^*$ is a minima of $L$ iff $\mathcal{B}(t^*)$ is a minima of $L'$.

 \begin{lemma}
 \label{lem: minima mapping}
     $t^*$ is a minima of $L$ iff $\mathcal{B}(t^*)$ is a minima of $L'$.
 \end{lemma}
\begin{proof}
    Let us first prove the forward direction. If $t^*$ is a local minima of $L$, then by definition, there exists an $\epsilon$ such that for all $t$, $|t-t^*|< \epsilon$, so we have that $L(t) \geq L(t^*)$. Set $\epsilon' = \epsilon$. We want to show that in the neighborhood around $a^* = \mathcal{B}(t^*)$, for all $a$, $L'(a) \geq L'(a^*)$.

    For the sake of contradiction, assume that there exists $a\in \mathcal{D'}$ such that $|a-a^*| \leq \epsilon'$ and $L'(a) < L'(a^*)$. Consider $t = \mathcal{B}^{-1}(a)$. First, observe that $|t-t^*|\leq \epsilon$ since $a$ has $t$ embedded in it. Also, we have $L(t) = L'(a) < L'(a^*) = L(t^*)$, which is a contradiction to the fact that $t^*$ is a local minima.

    We now prove the other direction. Let $a^*$ be a local minimum for $L'$. Hence, by definition, there exists an $\epsilon$ such that for all $a, |a-a^*|< \epsilon$, we have $L(a) \geq L(a^*)$. Let us set $\epsilon' = \epsilon/2n^2$. 

    For the sake of contradiction assume there exists a $t\in \mathcal{D}$ such that $|t-t^*|\leq \epsilon'$ and $L(t) < L(t^*)$. Observe that $(|\mathcal{B}(t) - \mathcal{B}(t^*)|)\leq 2\cdot \text{max}_i(t_i-t^*_i)\cdot n^2 \leq 2n^2\epsilon'$ (using the fact that the magnitude of a vector is more than the value of each of its components). Also, $(L'(\mathcal{B}(t) - L'(\mathcal{B}(t^*))< 0$, which is a contradiction to the fact that $a^*$ was a local minima. 
\end{proof}

Putting everything together, we have the following theorem.

\begin{theorem}
\label{thm: unique-minima}
        The function $L$ admits a unique minima in $\mathcal{D}$.
    \end{theorem}
\begin{proof}
     First, we showed that the function $L'$ admits a unique minima in $\mathcal{D}'$ in ~\cref{lem:unique minima for L'}. We then proved that the number of minima of $L$ in $\mathcal{D}$ is the same as the number of minima of $L'$ in  $\mathcal{D}'$ in ~\cref {lem: minima mapping}. Hence, $L$ admits a unique minima in $\mathcal{D}$.
\end{proof}

 As we have shown that $L$ has a unique minima. To find the minima, we first argue that all $T_i$'s can be expressed in terms of $T_1$.

\subsection{Dependency of $T_i$'s on $T_1$}\label{sec:decoupling}
 In this section, and in ~\cref {sec:approx}, we focus on the case $a = 1$, and defer the general case ($a \neq 1$) to ~\cref{sec:assume}. When $a = 1$, the two conditions in Line 2 of ~\cref{alg:aT_a getting} simplify to $\frac{1}{2^{n-2}} > \delta^T$ and $\frac{(\delta^T)^n}{(\delta^T + 1)^{n-2}} < \delta^T$. The second condition always holds, while the first condition holds only when $T > n \log_{1/\delta}(2)$. Therefore, in this section and in ~\cref{sec:approx}, we assume $T > n \log_{1/\delta}(2)$ whenever needed. The case $T \le n \log_{1/\delta}(2)$, where $a \neq 1$, is handled in ~\cref{sec:assume}.
 
 Let $T_i=\delta^{t_i}$ for $i \in \{0,1,\dots,n\}$. In this section, we show that $T_2, T_3 \dots, T_n$ can be expressed in terms of $T_1$. As shown in  \cref{sec:uniqueminima}, we know that $L=\sum_{j>i} \delta^{t_j-t_i}$ has a unique minima and therefore our eventual goal is to find this minima. The relation of $T_2, T_3, \dots, T_n$ in terms of $T_1$ will help us to find the minima of $L$. To this end, we show how $T_i$ is can be expressed in terms of $T_1,\dots, T_{i-1}$ and $T_{i+1}$ to $T_n$.

 \begin{lemma}\label{lem:quardatic}
 For $1\le i\le n-1$, for any strategy $\tee = (0, t_1, \cdots, t_{n-1}, T)$ corresponding to which $\nabla L = 0$, the following relations hold:\[(\delta^{t_i})^2=\left(\frac{1}{\frac{1}{\delta^{t_0}} + \dots + \frac{1}{\delta^{t_{i-1}}}}\right).(\delta^{t_{i+1}} + \dots + \delta^{t_n})\]
\end{lemma}
\begin{proof}
     As we know $\nabla L = 0$, therefore the partial derivative with respect to $t_1$ can  be written as 

   $$\frac{\partial(L)}{\partial(t_1)} = 0$$
   This implies that 
   $$\frac{\powdelta{t_1}}{\powdelta{t_0}}-\frac{1}{\powdelta{t_1}}(\powdelta{t_2}+\powdelta{t_3}+\cdots +\powdelta{t_n})=0$$
   We can rewrite the above equation as 
   $$(\powdelta{t_1})^2=\left(\frac{1}{\frac{1}{\powdelta{t_0}}}\right)\cdot (\powdelta{t_2}+\powdelta{t_3}+\cdots +\powdelta{t_n})$$

Similarly, for any $i\in \{2, \cdots, n-1\}$, we have

   $$\frac{\partial(L)}{\partial(t_{i})} = 0$$
   $$\powdelta{t_{i}}\left(\frac{1}{\powdelta{t_0}}+\frac{1}{\powdelta{t_1}}+\cdots +\frac{1}{\powdelta{t_{i-1}}}\right)-\frac{1}{\powdelta{t_{i}}}(\powdelta{t_{i+1}} + \dots +\delta^{t_n})=0$$
   \begin{equation}\label{eq:T_idefined}
       (\powdelta{t_{i}})^2=\left(\frac{1}{\frac{1}{\powdelta{t_0}}+\frac{1}{\powdelta{t_1}}+\cdots +\frac{1}{\powdelta{t_{i-1}}}}\right)\cdot (\powdelta{t_{i+1}}+ \dots + \delta^{t_n}).
   \end{equation}

This proves the lemma.
\end{proof}

Let $H_0=\left(\frac{1}{\delta^{t_0}}\right), H_1=\left(\frac{1}{\delta^{t_0}} + \frac{1}{\delta^{t_1}}\right)$, similarly $H_i=\left(\frac{1}{\delta^{t_0}} + \frac{1}{\delta^{t_1}} + \dots + \frac{1}{\delta^{t_i}} \right)=\sum_{j=0}^i \frac{1}{\delta^{t_j}}$. In the following lemma, we show how $T_i$ can be expressed in terms of $T_{i-1}$, $H_{i-1}$, $H_{i-2}$, which will later help us to express $T_i$ in terms of $T_1$.

\begin{lemma}\label{lem:Ti}
   For $1\le i\le n-1$, for any strategy $\tee = (0, t_1, \cdots, t_{n-1}, T)$ corresponding to which $\nabla L = 0$, the following relations hold:
    \begin{equation}
    \label{Eq: Decoupling}
        T_i = \frac{-1+\sqrt{1+4H_{i-2}H_{i-1}(T_{i-1}^2)}}{2H_{i-1}}
    \end{equation}
\end{lemma}
\begin{proof}
    As $T_i=\delta^{t_i}$, using \cref{lem:quardatic} we have 

\begin{equation}
    T_i^2 = \frac{1}{H_{i-1}}\cdot \left(\sum_{j=i+1}^nT_j\right)
\end{equation}

We can rewrite the above equation as 
\begin{equation}\label{eq:ti}
    H_{i-1}T_i^2 = \sum_{j=i+1}^nT_j
\end{equation}

Similarly, we have

\begin{equation}\label{eq:ti-1}
    H_{i-2}T_{i-1}^2 = \sum_{j=i}^nT_j
\end{equation}

Subtracting (\ref{eq:ti-1}) from (\ref{eq:ti}), we get

    $$H_{i-1}T_i^2 - H_{i-2}T_{i-1}^2 = \sum_{j=i+1}^nT_j - \sum_{j=i}^nT_j$$
    This implies
    $$H_{i-1}T_i^2 + T_i - H_{i-2}T_{i-1}^2 =0$$

Observe that we now have a quadratic equation in $T_i$. Also note that $T_i > 0$, as $\powdelta{t_i}>0$, the quadratic equation above can only have a positive root. Therefore,
\begin{equation*}
    T_i = \frac{-1+\sqrt{1+4H_{i-2}H_{i-1}(T_{i-1}^2)}}{2H_{i-1}}
\end{equation*}
\end{proof}

We now want to show that every $T_i$ can be expressed in terms of $T_1$. To this end, we first show that $T_2$ can be expressed as a function of $T_1$ and then formalize it for any $T_i$.

\begin{claim}\label{claim:t2}
    At any time $(t_1, \cdots, t_{n})$ where $\nabla L = 0$, $T_2$ can be written in terms of $T_1$ as follows:
    \begin{equation}
        T_2 = \frac{T^2_1}{1+T_1}.
    \end{equation}
\end{claim}
\begin{proof}

    From  \cref{lem:Ti}, we have 
\begin{equation*}
    T_i = \frac{-1+\sqrt{1+4H_{i-2}H_{i-1}(T_{i-1}^2)}}{2H_{i-1}}
\end{equation*} For $i=2$ we get

\begin{equation*}
    T_2 = \frac{-1+\sqrt{1+4H_{0}H_{1}(T_{1}^2)}}{2H_{1}}
\end{equation*}

 As $t_0=0$, we have $\delta^{t_0}=1$ which implies $H_0=1$ and $H_1=\left(1+ \frac{1}{T_1}\right)$. Putting this in the above equation, we get:
    
    \begin{align*}
        T_2 &= \frac{-1 + \sqrt{1+4\cdot 1\cdot (1+\frac{1}{T_1})T_1^2}}{2\cdot(1+\frac{1}{T_1})}\\
        &= T_1\frac{-1 + \sqrt{4T_1^2 + 4T_1 + 1}}{2\cdot(T_1+1)}\\
        &= T_1 \frac{-1 + 2T_1 + 1}{2(T_1 + 1)}\\
        &= \frac{T_1^2}{1+ T_1}.
    \end{align*}
\end{proof}

We now use \cref{lem:Ti} to show that every $T_i$ can be expressed in terms of $T_1$, which is the main theorem of this section.

\begin{theorem}
\label{thm:T_i expression}
    For any strategy $\tee = (0, t_1, \cdots, t_{n-1}, T)$ corresponding to which $\nabla L = 0$, for all $i$ in $\{1, \cdots, n-1\}$, $T_i$ can be written in terms of $T_1$ as follows:
   \begin{equation}
     T_i = \frac{T_1^i}{(1+T_1)^{i-1}}
    \end{equation}
\end{theorem}
\begin{proof}
      We prove this theorem by induction. As our base case, we know that the theorem holds for $T_2$ (\cref{claim:t2}).

    Let us assume that the induction hypothesis holds till $T_{i-1}$. We want to show that the theorem also holds for $T_{i}$. To start with, let us first calculate $H_{k}$ for $k<i$,
    \begin{align*}
        H_{k} &= \sum_{j=0}^{k} \frac{1}{T_j}\\
        &= 1+\sum_{j = 1}^{k}\frac{(1+T_1)^{j-1}}{T_1^j} \hspace{2cm}\text{ (Using induction hypothesis)}\\
        &= 1+ \frac{1}{T_1}\cdot \frac{(\frac{1+T_1}{T_1})^{k}-1}{\frac{1+T_1}{T_1}-1}\hspace{1.9cm}\text{(Using geometric progression formula)}\\
        &= 1+\frac{1}{T_1}\cdot \frac{(\frac{1+T_1}{T_1})^{k}-1}{\frac{1}{T_1}}\\
        &= 1+ \left(\frac{1+T_1}{T_1}\right)^{k}-1\\
        &= \left(\frac{1+T_1}{T_1}\right)^{k}.
    \end{align*}

    Using ~\cref{Eq: Decoupling}, we have,
    \begin{align*}
        T_{i} &= \frac{-1+\sqrt{1+4H_{i-2}H_{i-1}(T_{i-1}^2)}}{2H_{i-1}}\\
        &= \frac{-1 + \sqrt{1+4\cdot (\frac{(1+T_1)}{T_1})^{i-2}\cdot (\frac{(1+T_1)}{T_1})^{i-1}\cdot (\frac{T_1^{i-1}}{(1+T_1)^{i-2}})^2}}{2(\frac{1+T_1}{T_1})^{i-1}}\\
        &= T_1^{i-1}\cdot \frac{-1 + \sqrt{1 + 4(T_1+1)(T_1)}}{2(1+T_1)^{i-1}}\\
        &= T_1^{i-1}\cdot \frac{-1 +2T_1 + 1}{2(1+T_1)^{i-1}}\\
        &= \frac{T_1^{i}}{(1+T_1)^{i-1}}.
    \end{align*}
\end{proof}
\begin{corollary}
    At any time $(t_1, \cdots, t_n)$ where $\nabla L = 0$, for all $i$ in $[n-1]$, $H_i$ can be expressed in terms of $T_1$ as follows:
    \begin{equation}
        H_{i} = \left(\frac{1+T_1}{T_1}\right)^{i}
    \end{equation}
\end{corollary}

Note that in \cref{thm:T_i expression}, we have shown that $T_2, \dots, T_{n-1}$ can be expressed in terms of $T_1$. We now show how to express $T_n$ in terms of $T_1$.

\begin{lemma}\label{lem:T_nexpression}
    For $1\le i\le n$, for any strategy $\tee = (0, t_1, \cdots, t_{n-1}, T)$ corresponding to which $\nabla L = 0$, the following relation holds:
    \begin{equation}
        \frac{T_1^n}{(1+T_1)^{n-2}} = T_n
    \end{equation}
\end{lemma}
\begin{proof}
    From \cref{lem:quardatic}, we have that
 
\begin{equation}\label{eq:t12}
    T_1^2 = (T_2+T_3+\cdots + T_n)
\end{equation}

Which can be written as 
\[
T_1^2-(T_2+T_3+\cdots +T_{n-1})=T_n
\]

By using \cref{thm:T_i expression}, we substitute $T_i = \frac{(T_1)^i}{(1+T_1)^{i-1}}$ for $1<i<n$, to get

\begin{align*}
    T_2 + T_3 + \cdots + T_{n-1}&=\frac{(T_1)^2}{(1+T_1)}+ \frac{(T_1)^3}{(1+T_1)^2}+\cdots +\frac{(T_1)^{n-1}}{(1+T_1)^n}\\
    &= \frac{(T_1)^2}{(1+T_1)}\cdot \frac{1-(\frac{T_1}{1+T_1})^{n-2}}{1-\frac{T_1}{1+T_1}}\\
    &= T_1^2- \frac{T_1^n}{(1+T_1)^{n-2}}
\end{align*}

Hence, we have
\begin{equation*}
   \frac{T_1^n}{(1+T_1)^{n-2}} = T_n.
\end{equation*}
\end{proof}
\begin{corollary}
\label{cor: t_1 wrt t_n}
    For $1\le i\le n$, at any time $(t_1, \cdots, t_{n})$ where $\nabla L = 0$, the following relation holds:
    \begin{equation}
        \frac{\delta^{t_1n}}{(1+\delta^{t_1})^{n-2}} = \delta^T.
    \end{equation}
\end{corollary}

To conclude, we derived many interesting relations between $T_1,T_2,\dots,T_n$ in this section. In the following section, we use them to find an approximate value of $T_i$ and $t_i \forall i \in \{1,\dots,n-1\}$.
\subsection{Approximating $t_i$'s}\label{sec:approx}

In \cref{sec:uniqueminima}, we showed that a unique minima exists for $L$, and in \cref{sec:decoupling}, we showed that  $(T_2,\dots, T_n)$ can be expressed in terms of $T_1$. In this section, we show how to find an approximate solution for $(T_2,\dots, T_n)$ by finding an approximate solution for $T_1$. Note that as $T_1=\delta^{t_1}$, once we obtain an approximate solution for $T_1$, we also get a solution for $t_1$. To this end, we describe the outline of this section before proving everything in detail.

\begin{enumerate}
 \item We first show that $T_1$ has a unique solution.   
\item We then demonstrate how to approximate $T_1$ using binary search. Once we obtain an approximation for $T_1$, we use it to approximate $T_i$, $i\in \{2, \cdots, n\}$ using \cref{thm:T_i expression}.

    \item Next, we approximate $t_1$ using $T_1$, which further helps in approximating $t_2,\dots,t_n$.
\end{enumerate}

We first show that $T_1$ has a unique solution.
\subsubsection{$T_1$ has a unique solution}
In \cref{lem:T_nexpression}, we showed a relation between $T_1$ and $T_n$. Here we first show that the function $ \frac{T_1^n}{(1+T_1)^{n-2}}$ with respect to $T_1$ is strictly monotone and then argue that $T_1$ has a unique solution which minimizes $L$ provided that $T> n\log_{1/\delta}(2)$.

\begin{lemma}\label{lem:strictlyinc}
    Let $g(T_1)= \frac{T_1^n}{(1+T_1)^{n-2}}$. Then $g$ increases strictly monotonically with $T_1$.
\end{lemma}
\begin{proof}
    \begin{align*}
     \frac{d}{dT_1}\frac{T_1^n}{(1+T_1)^{n-2}} &=   \frac{nT_1^{n-1}}{(1+T_1)^{n-2}} - (n-2)\frac{T_1^n}{(1+T_1)^{n-1}}\\
     &=\frac{nT_1^{n-1}+2T_1^n}{(1+T_1)^{n-1}}\\
     &> 0
    \end{align*}
    Hence, $g(T_1)$ strictly monotonically increasing with $T_1$.
\end{proof}

Note that as $g(T_1)$ strictly monotonically increasing with $T_1$ and $T_n=\delta^T$ is a constant, we show that $T_1$ has a unique solution. To show this, for convenience, we assume that $T> n\log_{1/\delta}(2)$. We remove this assumption later in \cref{sec:assume}.

\begin{lemma}\label{lem:T>}
    Assuming $T> n\log_{1/\delta}(2)$, $g(T_1) = \delta^T$ admits a unique solution.
\end{lemma}
\begin{proof}
    Observe that $t_1\in \{0, \cdots,T\}$, so $T_1\in [\delta^T, 1]$.
    
    When $T_1 = \delta^T$,
    \begin{align*}
        g(T_1) &= \frac{T_1^n}{(1+T_1)^{n-2}}\\
        &< T_1\\
        & = T_n
    \end{align*}

    By our assumption that $T> n\log_{1/\delta}(2)$, let us show that $\delta^T < \frac{1}{2^{n-2}}$
    \begin{align*}
        T &> (n-2)\log_{1/\delta}(2)\\
        T\log(1/\delta) &> (n-2)\log(2)\\
        \delta^T &< \frac{1}{2^{n-2}}
    \end{align*}
    
    When $T_1 = 1$, 
    \begin{align*}
        g(T_1) &= 1/2^{n-2}\\
        &> T_n \hspace{1cm}\text{ (Assuming $T> n\log_{1/\delta}(2))$.}
    \end{align*}

    As $g(T_1) < T_n$ when $T_1 = \delta^T$ and $g(T_1) > T_n$ when $T_1 = 1$, using intermediate value theorem, there exists a point $p\in [\delta^T, 1]$ such that $g(p) = T_n$. Since $g$ is a continuous and strictly monotonically increasing function, the point $p$ is unique. Hence, $g(T_1)$ admits a unique solution.
\end{proof}

\begin{lemma}\label{lem:corresponding}
    A solution to $g(T_1) = \delta^T$ corresponds to a solution to $\nabla L = 0$.
\end{lemma}
\begin{proof}
    We began with a set of \( n \) independent equations (~\cref{eq:T_idefined}) and, through algebraic manipulation, derived another set of \( n \) independent equations (~\cref{cor: t_1 wrt t_n} and~\cref{thm:T_i expression}). Therefore, solving for \( g(T_1) = \delta^T \) and computing the corresponding values of \( T_i \) yields a solution that also satisfies the original equations. Since ~\cref{eq:T_idefined} implies \( \nabla L = 0 \), this solution corresponds to a minimum of \( L \).

\end{proof}

~\cref{lem:T>} states that when  $T > n\log_{1/\delta}(2)$, i.e., $T$ is much larger compared to $n$, there exists a unique interior point solution to our problem. In \cref{sec:assume}, we remove this assumption and show how to find the solution of $g(T_1)$. But for simplicity, we proceed with the assumption now. Next, we show how to approximate $T_1$ using binary search.

\subsubsection{Approximating $T_1$ using Binary Search}\label{sec:approxbinary}

As we know that $g(T_1)$ has a unique solution and $g$ is monotonically increasing with $T_1$, we perform binary search and estimate $T_1$ to an accuracy of $1/8^n$. 

\begin{lemma}\label{lem:binary}
     Assuming $T> n\log_{1/\delta}(2)$, we estimate the value of $T_1$ to an accuracy of $1/8^n$, in polynomial time.
\end{lemma}

\begin{proof}
    We know, $g(T_1) = \frac{T_1^n}{(1+T_1)^{n-2}}$ is strictly monotonically increasing with $T_1$. Let $$h(T_1) = g(T_1)-T_n.$$ As $T_n$ is a constant, $h(T_1)$ also strictly monotonically increasing with $T_1$. 

    Also, as $T_1\in [\delta^T, 1]$ and $T> n\log_{1/\delta}(2)$, it can be verified that at $T_1 = \delta^T$, $h(T_1) = h(\delta^T) < 0$, and at $T_1 = 1$, $h(T_1) = h(1) > 0$.

    We now perform binary search as described in ~\cref{alg:binary search}, to obtain $T_1$. We later show that \cref{alg:binary search} takes $O(3n+T\log_2{(1/\delta)})$ iterations to obtain an approximation of $1/8^n$. Since we only do constant-time operations and function evaluations in each iteration, we need $O(\ln(nT))$ time per iteration. Therefore, the overall time to obtain an approximation of $T$ to a factor of $1/8^n$ is $O((3n+T\log_2{(1/\delta)})\ln(nT))$.

    Note that the size of our search domain is initially $1-\delta^T < 1$ and in each round, the size of our domain decreases by a factor of $2$. Let $r$ be the total number of iterations, then after $r = 3n+T\log_2{(1/\delta)}$ iterations, our domain size decreases to a factor of less than
    \begin{align*}
        \frac{1}{2^{r}} &= \frac{1}{2^{3n+T\log_2{(1/\delta)}}}\\
        &= \frac{1}{8^n}\cdot {\delta^T}\\
        &\leq \frac{1}{8^n}\cdot T_1 \hspace{1cm} \text{($T_1 \geq \delta^T$)}
    \end{align*}

\end{proof}

\begin{algorithm}
    \caption{Binary Search}
    \label{alg:binary search}
    \begin{algorithmic}[1]
    \STATE $low \gets \delta^T$
    \STATE $high \gets 1$
    \STATE $n_{rounds} \gets 3n+T\log{(1/\delta)}$
    \STATE $mid \gets (low + high)/2$
    \STATE $curr \gets 1$
    \WHILE{$curr \leq n_{rounds}$}
        \STATE $mid \gets  (low + high) / 2 $
        \IF{$h(mid) = 0$}
            \RETURN $mid$
        \ELSIF{$h(mid) < 0$}
            \STATE $low \gets mid$
        \ELSE
            \STATE $high \gets mid$
        \ENDIF
        \STATE $curr \gets curr + 1$
    \ENDWHILE
    \RETURN $mid$
    \end{algorithmic}
    \end{algorithm}

For an optimal $t_1^*$, let $T_1^*=\delta^{t_1^*}$ be the optimal value of $T_i$ that satisfies $\nabla L = 0$. Similarly, define $T_2^*,\dots,T_n^*$. In the following lemma, we bound the error in $T_i$'s with respect to the estimation we make for $T_1$.

\begin{lemma}
\label{lem:approxt_n}
    Assuming $T_1^*(1-\epsilon) \leq T_1 \leq T_1^*(1+\epsilon)$, then $T_i$ can be bounded by  \[
T_i^*\cdot (1-4^n\epsilon)\leq T_i \leq T_i^*\cdot (1+4^n\epsilon)
\]
\end{lemma}
\begin{proof}
        Using \cref{thm:T_i expression} we know,
\begin{equation*}
    T_i = \frac{T_1^i}{(1+T_1)^{i-1}}
\end{equation*}

Let us assume that we can bound $T_1$ as
\[
T_1^*(1-\epsilon)\leq T_1 \leq T_1^*(1+\epsilon)
\]

Then we get,
\begin{align*}
    T_i &= \frac{T_1^i}{(1+T_1)^{i-1}}\\
    &\leq \frac{(T_1^*(1+\epsilon))^i}{(1+T_1^*(1-\epsilon))^{i-1}}\\
    &\le \frac{(T_1^*(1+\epsilon))^i}{((1+T_1^*)(1-\epsilon))^{i-1}}\\
    &= \frac{(T_1^*)^i}{(1+T_1^*)^{i-1}}\cdot \frac{(1+\epsilon)^i}{(1-\epsilon)^{i-1}}\\
    &=T_i^*\cdot \frac{(1+\epsilon)^i}{(1-\epsilon)^{i-1}}
\end{align*}

Using some known inequalities like $(1+x)^r\leq 1+(2^r-1)x$ and $\frac{1}{1-x} \leq 1+2x$ for $\epsilon < 1/2$.

We have, 
\begin{align*}
    T_i &\leq T_i^*\cdot \frac{(1+\epsilon)^i}{(1-\epsilon)^{i-1}}\\
    &\leq T_i^*\cdot (1+\epsilon)^i(1+2\epsilon)^{i-1}\\
    &\leq T_i^*\cdot (1+2\epsilon)^{2i-1}\\
    &\leq T_i^*\cdot (1+(2^{2i-1}-1)2\epsilon)\\
    &\leq T_i^*\cdot (1+(2^{2i}-1)\epsilon)
\end{align*}

Similarly, we get 
\begin{align*}
    T_i &= \frac{T_1^i}{(1+T_1)^{i-1}}\\
    &\geq \frac{(T_1^*(1+\epsilon))^i}{(1+T_1^*(1+\epsilon))^{i-1}}\\
    &\ge \frac{(T_1^*(1+\epsilon))^i}{((1+T_1^*)(1+\epsilon))^{i-1}}\\
    &= \frac{(T_1^*)^i}{(1+T_1^*)^{i-1}}\cdot \frac{(1-\epsilon)^i}{(1+\epsilon)^{i-1}}\\
    &=T_i^*\cdot \frac{(1-\epsilon)^i}{(1+\epsilon)^{i-1}}
\end{align*}

Again, we know that $(1-x)^r\geq 1-(2^r-1)x$ and $\frac{1}{1+x} \geq 1-x$.

We have, 
\begin{align*}
    T_i &\geq T_i^*\cdot \frac{(1-\epsilon)^i}{(1+\epsilon)^{i-1}}\\
    &\geq T_i^*\cdot (1-\epsilon)^i(1-\epsilon)^{i-1}\\
    &= T_i^*\cdot (1-\epsilon)^{2i-1}\\
    &\geq T_i^*\cdot (1-(2^{2i-1}-1)\epsilon)
\end{align*}

Putting everything together, we get  
\[
T_i^*\cdot (1-(2^{2i}-1)\epsilon)\leq T_i \leq T_i^*\cdot (1+(2^{2i-1}-1)\epsilon)
\]
We can rewrite the above equation as
\[
T_i^*\cdot (1-4^i\epsilon)\leq T_i \leq T_i^*\cdot (1+4^i\epsilon)
\]

Since $i\leq n$
\[
T_i^*\cdot (1-4^n\epsilon)\leq T_i \leq T_i^*\cdot (1+4^n\epsilon)
\]

Hence, we have proved the lemma.
\end{proof}
\begin{corollary}
    Setting $\epsilon = 1/8^n$, we get,  $$T_i^*\cdot (1-1/2^n)\leq T_i \leq T_i^*\cdot (1+1/2^n).$$
\end{corollary}

The above corollary is obtained by putting $\epsilon=1/8^n$ because we have already shown in \cref{lem:binary} that $T_1$ can be estimated to a factor of $1/8^n$ using \cref{alg:binary search}.

We now use the above lemma to approximate $t_i$, which is our main result.

\subsubsection{Approximating $t_i$ using $T_i$}

As we have shown that we can estimate $T_1,\dots, T_n$ in terms of multiplicative error, and since $T_i=\delta^{t_i}$, we now bound $t_i$ in terms of additive error.

\begin{lemma}\label{lem:optimalti}
   For $\epsilon < 0.5$, $t_i = \log_{\delta}T_i$ is bounded by
\[
    t_i^*-\frac{2\ln(2)\epsilon}{\log(1/\delta)} \leq t_i \leq t_i^*  +\frac{2\ln(2)\epsilon}{\log(1/\delta)}
\]
\end{lemma}
\begin{proof}
    As $T_i$ satisfies \[
T_i^*(1-\epsilon)\leq T_i \leq T_i^*(1+\epsilon)
\] We have,
    \begin{align*}
        T_i &\leq T_i^*(1+\epsilon)\\
        \delta^{t_i} &\leq \delta^{t_i^*}(1+\epsilon)\\
        t_i &\geq t_i^*+\log_{\delta}(1+\epsilon)\\
        t_i &\geq t_i^* - \frac{1}{\log{1/\delta}}\cdot \log(1+\epsilon)\\
        t_i &\geq t_i^* - \frac{1}{\log{1/\delta}}\cdot (\epsilon)
    \end{align*}

Similarly, we have
    \begin{align*}
    T_i &\geq T_i^*(1-\epsilon)\\
        \delta^{t_i} &\geq \delta^{t_i^*}(1-\epsilon)\\
        t_i &\leq t_i^*+\log_{\delta}(1-\epsilon)\\
        t_i &\leq t_i^*-\frac{1}{\log(1/\delta)}\log(1-\epsilon)\\
        t_i &\leq t_i^* - \frac{1}{\log(1/\delta)}(-2\ln(2)\epsilon) \\
        t_i &\leq t_i^* + \frac{2\ln(2)\epsilon}{\log(1/\delta)}
    \end{align*}
    Hence, $
    t_i^*-\frac{2\ln(2)\epsilon}{\log(1/\delta)} \leq t_i \leq t_i^*  +\frac{2\ln(2)\epsilon}{\log(1/\delta)}$.
\end{proof}

\begin{corollary}\label{cor:epsilon}
   
    Setting $\epsilon=\frac{1}{2^n}\cdot \frac{\log(1/\delta)}{2\ln(2)}$ we have,

    $$t_i^*-\frac{1}{2^n}\le t_i \le t_i^*+ \frac{1}{2^n}.$$

\end{corollary}

We have thus shown that each $t_i$ can be approximated within an exponentially small error of the optimal $t_i^*$, completing the analysis of our near-optimal strategy for $a = 1$. We next discuss several important implications and behaviors of the solution.

\section{Analysis when $a\neq 1$}\label{sec:assume}
In our analysis so far, we have assumed that $T> n\log_{1/\delta}(2)$, i.e., $a=1$. We now study the case when $T\le  n\log_{1/\delta}(2)$, i.e., $a\neq 1$. In this case,   the minima do not lie in the interior of space $\mathcal{D}$. However, in \cref{thm: unique-minima}, we have shown that there exists a unique minimum in $\mathcal{D}$. Therefore, such a minimum must lie on the boundary of $\mathcal{D}$.

To find the minima when $T \le n\log_{1/\delta}(2)$, we divide our analysis into three parts.
\begin{enumerate}
    \item We first discuss the set of inequalities for which the minima lie at the boundary of $\mathcal{D}$.
    
    \item Next, similar to \cref{sec:decoupling} we describe the relation between the $T_i$'s when $T \le n\log_{1/\delta}(2)$, which will eventually help us find the minima.
    \item  We then provide an algorithm to iteratively check the existence of the minima, if it exists, and how to find it.
\end{enumerate}

\subsection{Finding the right subspace to search for the minima}
This section focuses on determining the value of $a$. In our algorithm, we perform a linear search over possible values of $a$, and in each iteration, we check whether the conditions $$
\frac{a^{n - 2a}}{(a + 1)^{n - 2a}} \geq \delta^T \quad \text{and} \quad \frac{1}{a^2} \cdot \frac{(a\delta^T)^{n + 2 - 2a}}{(a\delta^T + 1)^{n - 2a}} \leq \delta^T $$

are satisfied. \cref{thm:domain} shows that this linear search correctly identifies the value of $a$, i.e., if $a$ satisfies these conditions, then in the optimal schedule, we have $t_i = 0$ and $t_{n - i} = T$ for all $i < a$.

From the definition of space $\mathcal{D}$ in \cref{sec:uniqueminima}, we note that the boundaries of $\mathcal{D}$ are $t_i \geq 0$, $t_i \leq T$ and $t_i \leq t_{i+1}$. But the set of equations obtained so far indicates that for the optimal $\tee$, $t_i<t_{i+1}$. We therefore show that any solution with the added constraint $t_i=t_{i+1}$ also corresponds to a solution to our initial set of equations for $\mathcal{D}$. To this end, with the additional constraint that $t_i=t_{i+1}$, we obtain the following two sets of equations.
    
    \begin{enumerate}
    
        \item Substituting $t_{i+1}$ for $t_{i}$ in $L$ and writing $\nabla L = 0$ where $\nabla$ is taken with respect to $\{t_1, \dots, t_{i-1}, t_{i+1}, \dots, t_{n-1}\}$.
        \item Substituting $t_{i}$ for $t_{i+1}$ in $L$ and writing $\nabla L = 0$ where $\nabla$ is taken with respect to $\{t_1, \dots, t_{i}, t_{i+2}, \dots, t_{n-1}\}$.
    \end{enumerate}
 We denote $\{t_1, \dots, t_{i-1}, t_{i+1}, \dots, t_{n-1}\}$ as our first equation and $\{t_1, \dots, t_{i}, t_{i+2}, \dots, t_{n-1}\}$ as our second equation. We now show the following lemma.
 \begin{lemma}\label{lem:ti=ti1}
     Any solution to the set of equations obtained by solving $\nabla L=0$, under the additional assumption $t_i = t_{i+1}$, is also a solution to the original set of equations obtained by solving $\nabla L = 0$.
 \end{lemma}

\begin{proof}
    We here present a proof sketch as the proof is similar to the proof of \cref{lem:quardatic}. 

    Let us denote our original set of equations obtained by $\partial(L)/\partial t_j$ in \cref{lem:quardatic} by $\mathcal{E}_j$. When we perform $\partial L/\partial t_j$, in our first set of equations, it looks like we have substituted $t_i$ with $t_{i+1}$ in our set of original equations. Denote by $\mathcal{E}^1_j$ the equation obtained by $\partial (L)/\partial t_j$ in our first set of equations. For example, when $\partial L/\partial t_1 = 0$, we get
    \[
    \frac{\partial(L)}{\partial t_1} = \delta^{t_1}(\frac{1}{\delta^{t_0}})+(\frac{1}{\delta^{t_1}})(\delta^{t_2}+\dots+2\delta^{t_{i+1}}+\dots+\delta^{t_n}).
    \]
    
    Similarly, our second set of equations looks like we have substituted $t_{i+1}$ with $t_{i}$ in our set of original equations.  Denote by $\mathcal{E}^2_j$ the equation obtained by $\partial (L)/\partial t_j$ in our second set of equations. For example, when $\partial L/\partial t_1 = 0$, we get
    \[
    \frac{\partial(L)}{\partial t_1} = \delta^{t_1}(\frac{1}{\delta^{t_0}})+(\frac{1}{\delta^{t_1}})(\delta^{t_2}+\dots+2\delta^{t_{i}}+\dots+\delta^{t_n}).
    \]
    
    Observe that since there is a unique minima for $L$, solving the above two equations must give the same solution (or it can be shown that both equations have no solution). We can now construct $\tee'=(t_1, \cdots, t_i, t_i, \dots, t_{n-1})$, where $\tee' = (t_1, \dots, t_{n-1})$ is the solution obtained by the first set of equations by calculating $\nabla L =0$, once substituting $t_{i+1}$ by $t_i$.

    Let us first talk about the constraints that are obtained by equation $\mathcal{E}_j^1$, where $j \neq i+1$. In our solution $\tee'$, we have that $t_i = t_{i+1}$. But $\mathcal{E}^1_j$ is identical to what we obtain when we set $t_i = t_{i+1}$ in $\mathcal{E}_j$. Hence, if we start with a solution $\tee'$ of $\mathcal{E}^1_j$, it will also be a solution to $\mathcal{E}_j$.

    Let $\tee''=(t_1, \cdots, t_{i+1}, t_{i+1}, \dots, t_{n-1})$, where $(t_1, \dots, t_{n-1})$ is the solution of the second equation. We observe that the solutions that we get by solving the first and second sets of equations must correspond to the minima, and we can easily show that this maps to the minima of $L$ by setting $t_{i+1} = t_i$ or $t_{i} = t_{i+1}$. Since there is a unique minima of $L$, $\tee' = \tee''$. 

    Let us now see the constraints that are obtained by equation $\mathcal{E}_j^2$, where $j \neq i+1$. In our solution $t'$, we have that $t_i = t_{i+1}$. But $\mathcal{E}^2_j$ is identical to what we obtain when we set $t_i = t_{i+1}$ in $\mathcal{E}_j$. Hence, if we start with a solution $t''$ of $\mathcal{E}^2_j$, it will also be a solution to $\mathcal{E}_j$.

    Putting these together, we obtain that $t'$ satisfies all the equations for our original set of equations.
\end{proof}

We now show that an analogous result to \cref{lem:ti=ti1} holds when we set $ t_i=0$ and $t_j = T$, where $j > i$.  Note that at the beginning we can set either $t_1 = 0$ or $t_{n-1} = T$. Observe that this is the most general thing to do, since if we set $t_i = 0$, where $i > 1$, then by the constraints, $t_j = 0$  $\forall j < i$, we will have $t_1 = 0$. Thus the added constraint still encompasses this solution. We now show that as $L$ has a unique optimum,  setting $t_1=0$ leads to $t_{n-1}=T$. We capture this in general by the following lemma.

 \begin{lemma}\label{lem:zero}
     If the optimal solution exists in the subspace $\mathcal{A}$ of $\mathcal D$, obtained by setting $t_1, t_2, \dots, t_i$ to zero and $t_{n-j}, \dots, t_{n-1}$ to zero. Then, there is an optimal solution in the subspace $\mathcal{A}'$ of $\mathcal{D}$ by setting $t_1, t_2, \dots, t_k$ to zero and $t_{n-k}, \dots, t_{n-1}$ to zero, where $k =
     \max(i, j)$.
 \end{lemma}

 \begin{proof}
     We show this by contradiction. Let us assume without loss of generality that $i< j$. Assume that a solution exists in the subspace $\mathcal{A}$, but not in $\mathcal{A}'$. Then this solution has a non-zero entry among $(t_{i}, \cdots, t_k)$, say at the $x^{th}$ position. 
     Let us call this solution $\tee = (t_1, t_2, \dots, t_{n-1})$, where $t_1 = \dots = t_i = 0$ and $t_{n-j} = \dots = t_{n-1} = 0$.

     For a constant $T$, let $T-\tee = (T-t_1, \dots, T-t_n)$. We call this solution $\tee'=(T-t_1, \dots, T-t_n)$.

     Observe that $L(\tee') = L(\tee)$. Therefore, if $\tee$ is a global minimum to $L$, then $\tee'$ is also a global minimum to $L$. But this contradicts the existence of unique global minima.
 \end{proof}

Let $\mathcal{D}_i$ be the subspace obtained by setting $t_i = 0$ and $t_{n-i} = T$. Similarly, define $\mathcal{D}_{i+1}$. Then, we have the following theorem:

\begin{theorem}\label{thm:domain}
    Consider the subspaces $\mathcal{D}_1, \mathcal{D}_2, \mathcal{D}_3, \dots, \mathcal{D}_{n/2}$. For each $i\in \{1, \dots, n/2\}$, one of the following hold:
    \begin{enumerate}
        \item The solution lies in the interior of some $\mathcal{D}_j$, where $j \leq i$.
        \item The solution lies on the boundary of $\mathcal{D}_1,\dots, D_{i}$ and in $\mathcal{D}_{i+1}$. 
    \end{enumerate}
\end{theorem}
\begin{proof}
    Let us prove the theorem by induction on $i$. For $i=1$, we have demonstrated in \cref{sec:uniqueminima}, that the solution lies in $\mathcal{D}_1$. Hence, if it does not lie in the interior of $\mathcal{D}_1$, then it lies at the boundary of $\mathcal{D}_1$. In \cref{lem:ti=ti1}, we have shown that if the minima of $L$ only satisfy equations of the form $t_i=t_{i+1}$, then our original set of equations also has a solution. But this cannot be the case as we have seen that $\mathcal{E}_j$ ($t_i \le t_{i+1}$) has a solution. Hence, the minima for $L$ must also satisfy some equations of the form $t_j = 0$ or/and $t_k = T$ for some $j, k > 1$. However, if it satisfies some equation of the form $t_j=0$ for some $j>i$, then by \cref{lem:zero} it also satisfies $t_{2}=0$, and $t_{n-2} = T$ and so on.

    Now, let us assume that the theorem holds upto $i-1$ and demonstrate it for $i$. By the induction hypothesis, we know that either the solution lies in the interior of $\mathcal{D}_1, \dots, \mathcal{D}_{i-1}$ or it lies in $\mathcal{D}_i$. Let us consider the case where it lies in $\mathcal{D}_i$. If it does not lie in the interior of $\mathcal{D}_i$, then it lies on some boundary of $\mathcal{D}_i$. If the solution lies at the boundary of $\mathcal{D}_i$, then it must satisfy some of the equations that form the boundary. As the minima lies in $\mathcal{D}_i$, it satisfies $t_i = 0$ and $t_{n-i} = 0$, using \cref{lem:zero}.
    
    Again, by \cref{lem:ti=ti1}, we know that if the minima to $L$ only satisfy equations of the form $t_i=t_{i+1}$ and not any other equation, then our original set of equations has a solution. Hence, the minima of $L$ must satisfy some equations of the form $t_j = 0$ or/and $t_k = T$ for some $j, k > i$. However, if it satisfies the equation of the form $t_j=0$ for some $j>i$, it also satisfies $t_{i+1}=0$. Also using  \cref{lem:zero}, it must satisfy $t_{n-i-1} = T$. Similarly by \cref{lem:zero} if for some $j > i$, our equation satisfies $t_{n-i-1} = T$, then it also satisfies, $t_{i+1} = 0$.
    
    Therefore, by these inequalities, the minima must satisfy $t_{i+1}=0$ and $t_{n-i-1} = T$. This proves our lemma.
\end{proof}

In the next section, we show how the $T_i$'s depend on one another, similar to \cref{sec:decoupling}, where we expressed each $T_i$ in terms of $T_1$.

\subsection{Dependency of $T_i$}

In this section, we show how the $T_i$'s relate to one another under the new set of equations. While many of the proofs mirror those in ~\cref{sec:decoupling}, we now extend them to a more general setting by incorporating the constraint $t_a = 0$ for different values of $a$.

As we have already seen, if the solution of $L$ does not lie in the interior of $\mathcal{D}$, then it must lie in the subspace with the added constraints $t_a=0$ and $t_{n-a} = T$ for some $a\in \{1,2,\dots,n/2\}$. To this end, we show how $T_i$ can be expressed in terms of $T_a$.
 
 \begin{lemma}\label{lem:ta}
     Assuming $t_0 = \dots = t_{a-1} = 0$, we can write $T_{a+i}$ in terms of $T_a$ as 
     \begin{equation}
     \label{eq:gen t_i}
         T_{a+i} = \frac{1}{a}\cdot \frac{(aT_a)^{i+1}}{(aT_a+1)^{i}}.
     \end{equation}
 \end{lemma}
\begin{proof}
Using \cref{lem:Ti}, we have

    \begin{equation*}
        T_i = \frac{-1+\sqrt{1+4H_{i-1}H_{i-2}T_i^2}}{2H_{i-1}}
    \end{equation*}

    As a base case, when $i=0$ we have, $ \frac{1}{a}\cdot \frac{(aT_a)^1}{(aT_a+1)^{0}}=T_a$. 
    
    As the first $a$ terms are zero, using the induction hypothesis, we now calculate $H_{a+k}$, which is the sum of the first $a+k+1$ terms.
    \begin{align*}
        H_{a+k} &= \sum_{j=0}^{k} \frac{1}{T_j}\\
        &= a+a\sum_{j = 0}^{k}\frac{(1+aT_1)^{j}}{(aT_1)^{j+1}} \hspace{2.6cm}\text{(Using induction hypothesis)}\\
        &= a+a\sum_{j = 1}^{k+1}\frac{(1+aT_1)^{j-1}}{(aT_1)^{j}} \hspace{2cm}\\
        &= a+ \frac{1}{T_1}\cdot \frac{(\frac{1+aT_1}{aT_1})^{k+1}-1}{\frac{1+aT_1}{aT_1}-1}\hspace{2cm}\text{(Using geometric progression formula)}\\
        &= a+\frac{1}{T_1}\cdot \frac{(\frac{1+aT_1}{aT_1})^{k+1}-1}{\frac{1}{aT_1}}\\
        &= a+ a\left(\frac{1+aT_a}{aT_1}\right)^{k+1}-a\\
        &= a\left(\frac{1+aT_a}{aT_a}\right)^{k+1}
    \end{align*}

Now, we have
\begin{align*}
        T_{a+i} &= \frac{-1+\sqrt{1+4H_{a+i-2}H_{a+i-1}(T_{a+i-1}^2)}}{2H_{a+i-1}}\\
        &= \frac{-1 + \sqrt{1+4\cdot a(\frac{(1+aT_1)}{aT_1})^{i-1}\cdot a(\frac{(1+aT_1)}{aT_1})^{i}\cdot (\frac{1}{a}\frac{(aT_1)^{i}}{(1+aT_1)^{i-1}})^2}}{2a(\frac{1+aT_1}{aT_1})^{i}}\\
        &= (aT_1)^{i}\cdot \frac{-1 + \sqrt{1 + 4(aT_1+1)(aT_1)}}{2a(1+aT_1)^{i}}\\
        &= (aT_1)^{i}\cdot \frac{-1 +2aT_1 + 1}{2a(1+aT_1)^{i}}\\
        &= \frac{1}{a}\frac{(aT_1)^{i+1}}{(1+aT_1)^{i}}.
    \end{align*}

\end{proof}
Next, similar to \cref{lem:T_nexpression}, we show how to express $T_{n-a+1}$ in terms of $T_a$.
\begin{lemma}\label{lem:endcondition}
     Assuming $t_0 = \dots = t_{a-1} = 0$, we have
     \begin{equation}\label{eq:general t_n}
    T_{n-a+1}=\frac{1}{a^2}\frac{(aT_a)^{n-2a+2}}{(1+aT_a)^{n-2a}}.
\end{equation}
 \end{lemma}

\begin{proof}
    Similar to how we obtained ~\cref{eq:t12}, we use $\partial L/\partial t_a = 0$, and the fact that $t_0=\dots=t_{a-1}=0$ and $t_{n-a+1} = \dots = t_n = T$, we get the following

\[
T_{a}^2 = \frac{1}{a}(T_{a+1}+\dots+T_{n-a}+a\delta^T) 
\]
Hence,
\[
aT_{a}^2  - T_{a+1}+\dots+T_{n-a} = a\delta^T
\]
Now, using Lemma \ref{lem:ta} we have
\begin{align*}
    T_{a+1} + T_{a+2} + \cdots + T_{n-a}&=\frac{1}{a}(\frac{(aT_a)^2}{(1+aT_a)}+ \frac{(aT_a)^3}{(1+aT_a)^2}+\cdots +\frac{(aT_a)^{n-2a+1}}{(1+aT_a)^{n-2a}})\\
    &= \frac{1}{a}(\frac{(aT_a)^2}{(1+aT_a)}\cdot \frac{1-(\frac{aT_a}{1+aT_a})^{n-2a}}{1-\frac{aT_a}{1+aT_a}})\\
    &= a(T_a)^2- \frac{1}{a}\frac{(aT_a)^{n-2a+2}}{(1+aT_a)^{n-2a}} 
\end{align*}

Thus, putting everything together, we have 
\begin{equation}\label{eq:tan}
    \frac{(aT_a)^{n-2a+2}}{(1+aT_a)^{n-2a}} = a^2\delta^T.
\end{equation}
\end{proof}

\begin{corollary}
     Assuming $t_0 = \dots = t_{a-1} = 0$, we have
     \begin{equation}\label{eq:general t_n for t}
    \delta^T=\frac{1}{a^2}\frac{(a\delta^{t_a})^{n-2a+2}}{(1+a\delta^{t_a})^{n-2a}}.
    \end{equation}
\end{corollary}

Having derived the relation between $T_{a+i}$ and $T_a$, we next use them to approximate $T_a$ and provide our near-optimal algorithm.

\subsection{Checking minima and approximation}
In this section, we derive the two conditions used in ~\cref{alg:aT_a getting} to find the value of $a$. We then provide the fleshed-out version of ~\cref{alg:aT_a getting} in ~\cref{alg:binary search with checking} and derive our main theorem for the general setting ($a\neq1$).

Similar to \cref{sec:approx}, we first define 
\begin{equation}
    g(T_a)= \frac{1}{a^2}\cdot \frac{(aT_a)^{n-2a+2}}{(1+aT_a)^{n-2a}}
\end{equation}

\begin{lemma}
\label{lem: gen montonic increasing}
$g(T_a)$ is strictly monotonically increasing with $T_a\in [\delta^T, 1]$.
\end{lemma}
\begin{proof}
    \begin{align*}
     \frac{d}{dT_a}\frac{1}{a^2}\cdot \frac{(aT_a)^{n-2a+2}}{(1+T_a)^{n-2a}} &=   \frac{(n-2a+2)(aT_a)^{n-2a+1}}{(1+aT_a)^{n-2a}} - (n-2a)\frac{(aT_a)^{n-2a+2}}{(1+aT_a)^{n-2a+1}}\\
     &=\frac{(n-2a+2)(aT_a)^{n-2a+1}+2(aT_a)^{n-2a+2}}{(1+aT_a)^{n-1}}\\
     &> 0
    \end{align*}
\end{proof}

\begin{corollary}
     $g(T_a) = \delta^T$ admits at most one solution. 
\end{corollary}

We have shown in ~\cref{lem:endcondition} that any solution that satisfies ~\cref{eq:tan} is the optimum solution.
To check if ~\cref{eq:tan}  admits a solution, it suffices to check the value of $g(T_a)$ at $t_a = 0$ and $t_a = T$. Precisely, we need to check if
\begin{itemize}
    \item at $t_a = 0$
\[
\frac{a^{n-2a}}{(a+1)^{n-2a}} \geq \delta^T
\]
\item  at $t_a = T$
\[
\frac{1}{a^2}\cdot\frac{(a\delta^T)^{n+2-2a}}{(a\delta^T+1)^{n-2a}}\leq \delta^T
\]
\end{itemize}
The two above inequalities can be verified in \( O(\ln(nT)) \) time. According to ~\cref{thm:domain}, domain reduction, which involves checking whether the inequalities hold and incrementing \( a \) if they do not, needs to be performed at most \( n/2 \) times.
If we find that a feasible solution to these inequalities appears after performing domain reduction up to \( a - 1 \) times (for some \( a \leq n/2 \)), then, by ~\cref{thm:domain}, this solution lies in the interior of the domain \( D_a \).

If we have to do domain reduction $n/2$ times, then there is only one point in the feasible region i.e., $t_i = 0$ for $i \leq n/2$ and $t_i = T$ for $i > n/2$ (if $n$ is odd, then $t_{n/2} = T/2$).

Once we find a value of \( a \) that satisfies both the conditions, we perform binary search using ~\cref{alg:binary search with checking}, which runs in \( O\left((3n + T\log_2(1/\delta))\ln(nT)\right) \) time. This yields an approximate solution to \( T_a \) with accuracy \( 1/8^n \), following an analysis similar to that in ~\cref{lem:binary}. As in ~\cref{lem:corresponding}, we can show that a solution to $g(T_a) = \delta^T$ corresponds to a solution to $\nabla L = 0$, and thus a minimum to $L$. 

After computing \( T_a \), we can recover the value of \( T_i \)'s using ~\cref{lem:ta} in \( O(n\log(nT)) \) time. Applying a similar analysis as in ~\cref{lem:approxt_n}, we obtain all \( T_i \) with accuracy \( 1/4^n \). Finally, using these approximate values of \( T_i \), we compute the optimal schedule \( t_i \) in an additional \( O(n\ln(nT)\ln(1/\delta)) \) time, with an additive error of at most \( 
\frac{1}{2^n} \).

Note that ~\cref{alg:binary search with checking} is the complete description of ~\cref {alg:aT_a getting}. We now prove one of our main theorem.

\begin{algorithm}
    \caption{Binary Search with checking}
    \label{alg:binary search with checking}
    \begin{algorithmic}[1]
    \STATE $low \gets \delta^T$
    \STATE $high \gets 1$
    \STATE $n_{rounds} \gets 3n+T\log{(1/\delta)}$
    \STATE $mid \gets (low + high)/2$
    \STATE $curr \gets 1$
    \STATE $a\gets 1$
    \STATE $ctr \gets 1$
    \WHILE{$ctr \leq n/2$}
        \IF{$\frac{a^{n-2a}}{(a+1)^{n-2a}} < \delta^T$}
            \STATE $ctr++$
        \ELSIF{$\frac{1}{a^2}\cdot\frac{(a\delta^T)^{n+2-2a}}{(a\delta^T+1)^{n-2a}}> \delta^T$}
            \STATE $ctr++$
        \ELSE   
            \STATE $a\gets ctr$
            \STATE break
        \ENDIF
    \ENDWHILE
    \IF{$ctr\geq n/2$}
        \RETURN $(n/2, \bot)$
    \ENDIF
    \STATE $h(T_a)= \frac{1}{a^2}\cdot \frac{(aT_a)^{n-2a+2}}{(1+aT_a)^{n-2a}}$
    \WHILE{$curr \leq n_{rounds}$}
        \STATE $mid \gets  (low + high) / 2 $
        \IF{$h(mid) = 0$}
            \RETURN $mid$
        \ELSIF{$h(mid) < 0$}
            \STATE $low \gets mid$
        \ELSE
            \STATE $high \gets mid$
        \ENDIF
        \STATE $curr \gets curr + 1$
    \ENDWHILE
    \RETURN $(a, mid)$
    \end{algorithmic}
    \end{algorithm}

    \begin{theorem}
        From ~\cref{alg:binary search with checking}, we can derive a $1/2^n$-approximate schedule to the optimal schedule using ~\cref{alg: optimal schedule}.
    \end{theorem}

    \begin{proof}
        
    As shown in ~\cref{thm: unique-minima} of ~\cref{sec:uniqueminima}, there exists a unique optimal schedule. In ~\cref{sec:decoupling}, ~\cref{lem:Ti} and~\cref{lem:T_nexpression} establish a system of equations whose solution yields this optimal schedule. ~\cref{lem:optimalti} in ~\cref{sec:approx} shows that our algorithm provides a \(1/2^n\)-approximation to the solution of these equations. For the general case \(a \neq 1\), ~\cref{lem:ta} and~\cref{lem:endcondition} in ~\cref{sec:assume} present analogous expressions and we can observe that the remaining ads are equispaced between the first and last non-end ad, and we show that the algorithm achieves the same \(1/2^n\)-approximation in this broader context. Putting everything together, we conclude that our algorithm yields a \(1/2^n\)-approximation to the optimal schedule.  We next discuss several important implications and behaviors of the solution.

    \end{proof}

\section{Implications}\label{sec:implication}
Having already established the optimality for the objective function $L$, we now proceed to analyze its structural properties. These properties reveal how the placement of advertisements varies with $\delta$, offering deeper insights into the behavior of our strategy under different values of $\delta$. To
begin, we first present the following property of our solution.
\begin{lemma}\label{lem:firstlastgap}
    For the optimal solution $\tee$, we have $t_1-t_0 = t_n-t_{n-1}$.
\end{lemma}
\begin{proof}
    At optimum, we have $T_n = \frac{T_1^n}{(1+T_1)^{n-2}}$ and $T_{n-1} = \frac{T_1^{n-1}}{(1+T_1)^{n-2}}$. Hence, 
    
    \[\frac{T_n}{T_{n-1}} = T_1 = \frac{T_1}{T_0}\]

    The above equation can be written as
$$\delta^{t_n-t_{n-1}} = \delta^{t_1-t_0}$$

Therefore, we get

$$t_n-t_{n-1}=t_1-t_0$$

\end{proof}
Next, we show that the ratio between $T_i$ and $T_{i-1}$ is always fixed.

\begin{lemma}\label{lem:equalratio}
For $i \in \{2, \dots, n-1\}$, we have
\begin{equation}
    \frac{T_{i}}{T_{i-1}} = \frac{T_1}{(1+T_1)}.
\end{equation}    
\end{lemma}

\begin{proof}
    We have,
    \begin{align*}
    \frac{T_{i}}{T_{i-1}} &= \frac{\frac{T_1^{i}}{(1+T_1)^{i-1}}}{\frac{T_1^{i-1}}{(1+T_1)^{i-2}}}  \hspace{1cm}\text{Using \cref{thm:T_i expression}} \\
    &= \frac{T_1}{(1+T_1)}
    \end{align*}
\end{proof}
\begin{corollary}\label{cor:equalspace}
    $t_2-t_1 = t_3-t_2 = \dots = t_{n-1}-t_{n-2}$.
\end{corollary}
Now, in the following two lemmas, we show that as we increase the value of $\delta$, the value of $t_1$ decreases. We first show that $t_1$ cannot be greater than $T/n$.
\begin{lemma}
   Given $\delta, T  \text{ and } n$, in any optimal solution $\tee$ we have $t_1< {T/n}$ .
\end{lemma}
\begin{proof}
    From ~\cref{cor: t_1 wrt t_n}, we know that when $\nabla L = 0$, we have 
    \begin{equation}
        \frac{\delta^{t_1n}}{(1+\delta^{t_1})^{n-2}} = \delta^T
    \end{equation}

    Let us define $h(t_1) = \frac{\delta^{t_1n}}{(1+\delta^{t_1})^{n-2}}$. At optimum, we would like $h(t_1) = \delta^T$.
    
    For a fixed $\delta$, 
    
    \begin{align*}
        h(T/n) &= \frac{\delta^{n\frac{T}{n}}}{(1+\delta^\frac{T}{n})^{n-2}}\\
        &< \delta^T
    \end{align*}
    
   For a fixed $\delta, T \text{ and } n$, we know that $g(T_1) = \frac{T_1^{n}}{(1+T_1)^{n-2}}$ is strictly monotonically decreasing with $T_1$ (using \cref{lem:strictlyinc}). 

    We have that $h(t_1) = g(\delta^{t_1})$. Since $\delta^{t_1}$ monotonically decreases with $t_1$, we obtain that $h(t_1)$ monotonically decreases with $t_1$.
   
   We know that $h(T/n) <\delta^T$. Since in the optimal solution, we require $h(t_1) = \delta^T$, we obtain $t_1< T/n$.
\end{proof}

\begin{lemma}\label{lem:incdec}
    Fix $n$ and $T$, then as $\delta$ increases, the value of $t_1$ decreases.
\end{lemma}
\begin{proof}
We know from ~\cref{cor: t_1 wrt t_n}, that at optimal value of $\tee$, we have
    \begin{align*}
     \delta^T  &= \frac{\delta^{nt_1}}{(1+\delta^{t_1})^{n-2}}
    \end{align*}
    Differentiating the above equation from both sides, we get
    \begin{align*}
        \frac{T}{\delta}\frac{\delta^{nt_1}}{(1+\delta^{t_1})^{n-2}}&= \frac{(1+\delta^{t_1})^{n-2} \left(nt_1\delta^{nt_1-1}+n\delta^{nt_1}\ln(\delta)\frac{\partial t_1}{\partial \delta}\right)}{(1+\delta^{t_1})^{2n-4}} 
        &\\&- 
        \frac{\delta^{nt_1}(\left(n-2)(1+\delta^{t_1})^{n-3}t_1\delta^{t_1-1}+(n-2)(1+\delta^{t_1})^{n-3}\delta^{t_1}\ln(\delta)\frac{\partial t_1}{\partial \delta}\right)}{(1+\delta^{t_1})^{2n-4}}
    \end{align*}
    This implies that 
    \begin{align*}
        T(1+\delta^{t_1})&=(1+\delta^{t_1})(nt_1+n\delta\ln(\delta)\frac{\partial t_1}{\partial \delta}) - \delta((n-2)t_1\delta^{t_1-1}+(n-2)\delta^t\ln(\delta)\frac{\partial t_1}{\partial \delta})
    \end{align*}
    Hence, we have
    \begin{align*}
        T(1+\delta^{t_1})&= \frac{\partial t_1}{\partial \delta}((1+\delta^{t_1})n\delta \ln\delta - (n-2)\delta^{t+1}\ln(\delta))+nt_1+2t_1\delta^{t_1}
    \end{align*}
    Therefore,
    \begin{align*}
        (\frac{\partial t_1}{\partial \delta})(\delta \ln \delta) (n+2\delta^{t_1}) &= T(1+\delta^{t_1}) - (n+2\delta^{t_1})t_1
    \end{align*}

    Rearranging the terms, we get
    \begin{align*}
        (\frac{\partial t_1}{\partial \delta})(\delta \ln \delta) (n+2\delta^{t_1}) &= \delta^{t_1}(T-2t_1)+(T-nt_1)
    \end{align*}

    Finally, we obtain
    \begin{align*}
        \frac{\partial t_1}{\partial \delta} &= \frac{\delta^{t_1}(T-2t_1) + (T-nt_1)}{\delta\ln{\delta}(n+2\delta^{t_1})}
    \end{align*}

    As $n> 2$ and $t_1 < T/n$,  the numerator in the above equation is strictly positive. Note that $\ln(\delta) < 0$ and therefore the denominator is strictly negative. Hence, the derivative must also be strictly less than 0. This implies that $t_1$ is monotonically decreasing with $\delta$.
\end{proof}
Next, we show that as $\delta\rightarrow0$, $t_i$'s are equally spaced in $T$.
\begin{lemma}
    As $\delta\rightarrow 0$, we have $t_i \rightarrow Ti/n$.
\end{lemma}
\begin{proof}
    As $T_n=\delta^T$, for the optimum solution, the corresponding value of $T_1$ satisfies
    \begin{align*}
        \delta^T &= \frac{T_1^n}{(1+T_1)^{n-2}}
    \end{align*}

 But if the solution is not optimal, the above equation might not be true. Hence, we try to estimate the value of the quantity:
    \[
    \frac{1}{\delta^T}\frac{T_1^n}{(1+T_1)^{n-2}}
    \]

    We would like to show that as $\delta\rightarrow 0$, the value of the above equation at $T_1 = \delta^{T/n}$ converges to 1, which will imply that $\frac{T_1^n}{(1+T_1)^{n-2}}\rightarrow \delta^T$. We now show that this is true.
    \begin{align*}
        \lim_{\delta\rightarrow0}\frac{T_1^{n}}{\delta^T(1+T_1)^{n-2}} &=\lim_{\delta\rightarrow0}\frac{\delta^T}{\delta^T(1+\delta^T)^{n-2}}\\
        &=\lim_{\delta\rightarrow0}\frac{1}{(1+\delta^T)^{n-2}}\\
        &= 1
    \end{align*}

   As we have already seen that $t_i$'s are equally spaced between $t_1$ and $t_{n-1}$ ( \cref{cor:equalspace}) and since $t_1 \rightarrow T/n$ and $t_{n-1} \rightarrow T-T/n$ (using \cref{lem:firstlastgap}), as $\delta \rightarrow 0$, this gives $t_i \rightarrow iT/n$.
\end{proof}

The properties that we have shown so far in this section hold when $T >  n\log_{1/\delta}(2)$ (i.e., $a=1$). We now show that similar properties hold when $T \le  n\log_{1/\delta}(2)$ (i.e., $a \neq 1$). From ~\cref{sec:assume}, we observe that as $\delta$ increases, $t_{a-1}=0$ and $t_{n-a+1} = T$ for $a \in [n/2+1]$ hold.  We now, in the next three lemmas, argue that as $\delta$ increases, $t_a$ decreases for $a\neq 1$.

Firstly, from the analysis in ~\cref{sec:assume} we observe that as $\delta$ increases, the constraint $t_a = 0$ holds for larger and larger $a$. When $t_{a-1}=0$ and $t_a\neq 0$, from ~\cref{eq:general t_n for t}, we can write $
    \delta^T = \dfrac{1}{a^2}\dfrac{(a\delta^{t_a})^{n-2a+2}}{(1+a\delta^{t_a})^{n-2a}}
    $. We now try to find the ranges of $\delta$ for which  $
    \delta^T = \dfrac{1}{a^2}\dfrac{(a\delta^{t_a})^{n-2a+2}}{(1+a\delta^{t_a})^{n-2a}}
    $ holds.

\begin{lemma}
    The equation
    \begin{equation}
    \label{eq: ta in terms of T}
    \delta^T = \frac{1}{a^2}\frac{(a\delta^{t_a})^{n-2a+2}}{(1+a\delta^{t_a})^{n-2a}}    
    \end{equation}
    
    holds in the domain $\delta \in \left[\left(\dfrac{1}{(1+\frac{1}{a-1})^{n-2a+2}}\right)^{1/T}, \left(\dfrac{1}{(1+\frac{1}{a})^{n-2a}}\right)^{1/T}\right]$
\end{lemma}

\begin{proof}
    ~\cref{eq: ta in terms of T} holds when 
    \begin{itemize}
        \item Solving ~\cref{eq: ta in terms of T} gives a positive solution for $t_a$.
        \item Solving the equation 
        \begin{equation}
        \label{eq: ta-1 in terms of T}
            \delta^T = \frac{1}{(a-1)^2}\frac{((a-1)\delta^{t_a})^{n-2a+4}}{(1+(a-1)\delta^{t_{a-1}})^{n-2a+2}}
        \end{equation} 
        gives a negative solution for $t_{a-1}$, i.e., the equation is invalid.
    \end{itemize}

    Let us define $h_a(t_a) = \dfrac{1}{a^2}\dfrac{(a\delta^{t_a})^{n-2a+2}}{(1+a\delta^{t_a})^{n-2a}}$.
    
    Using \cref{lem: gen montonic increasing}, we know that $g(T_a)= \dfrac{1}{a^2}\cdot \dfrac{(aT_a)^{n-2a+2}}{(1+aT_a)^{n-2a}}$ is monotonically increasing. with $T_a$ 

    We know that $h_a(t_a) = g(\delta^{t_a})$. Since $\delta^{t_a}$ monotonically decreases with $t_a$, we have that $h_a(t_a)$ monotonically decreases with $t_a$.
    
    Substituting $t_a=0$, we get 
    \begin{align*}
        h(0) = \frac{1}{a^2}\frac{a^{n-2a+2}}{(1+a)^{n-2a}}\\
    \end{align*}

    We know that $h(t_a)\leq h(0)$ since $t_a\geq 0$ by monotonicity. We also know that if Equation~\ref{eq: ta in terms of T} has a solution $t_a$, we would have $h(t_a) = \delta^T$. Using these, we obtain that: 
    \begin{align*}
        \delta^T &\leq \frac{1}{a^2}\frac{a^{n-2a+2}}{(1+a)^{n-2a}}\\
        \delta &\leq \left(\frac{1}{(1+1/a)^{n-2a}}\right)^{1/T}
    \end{align*}

    Hence, ~\cref{eq: ta in terms of T} is valid only for $\delta \leq \left(\frac{1}{(1+1/a)^{n-2a}}\right)^{1/T}$.

    By a similar argument, we can show that ~\cref{eq: ta-1 in terms of T} is valid only for $\delta \leq \left(\frac{1}{(1+\frac{1}{a-1})^{n-2a+2}}\right)^{1/T}$. For $\delta\geq \left(\frac{1}{(1+\frac{1}{a-1})^{n-2a+2}}\right)^{1/T}$, we would obtain $t_{a-1} = 0$ and ~\cref{eq: ta-1 in terms of T} would no longer hold. Therefore, it becomes valid to now write ~\cref{eq: ta in terms of T} as long as $\delta \leq \left(\frac{1}{(1+1/a)^{n-2a}}\right)^{1/T}$.

    In the domain $\delta \in \left[\left(\frac{1}{(1+\frac{1}{a-1})^{n-2a+2}}\right)^{1/T}, \left(\frac{1}{(1+\frac{1}{a})^{n-2a}}\right)^{1/T}\right]$, we have that:
    \begin{itemize}
        \item The equation $\delta^T = \dfrac{1}{(a-1)^2}\dfrac{((a-1)\delta^{t_a})^{n-2a+4}}{(1+(a-1)\delta^{t_{a-1}})^{n-2a+2}}$ is invalid.
        \item The equation $\delta^T = \dfrac{1}{a^2}\dfrac{(a\delta^{t_a})^{n-2a+2}}{(1+a\delta^{t_a})^{n-2a}}$ is valid.
    \end{itemize}
    Hence, the domain in which \cref{eq: ta in terms of T} holds is $\delta \in \left[\left(\dfrac{1}{(1+\frac{1}{a-1})^{n-2a+2}}\right)^{1/T}, \left(\dfrac{1}{(1+\frac{1}{a})^{n-2a}}\right)^{1/T}\right]$. 
\end{proof}
 
\begin{lemma}
    \label{lem: val t_a at boundary}
    When $\delta = \left(\dfrac{1}{(1+\frac{1}{a-1})^{n-2a+2}}\right)^{1/T},$ we have $t_a = T/(n-2a)$ as the only solution to ~\cref{eq: ta in terms of T},    i.e. $\delta^T = \dfrac{1}{a^2}\dfrac{(a\delta^{t_a})^{n-2a+2}}{(1+a\delta^{t_a})^{n-2a}}$
    .
\end{lemma}

\begin{proof}
  We can verify the statement of the lemma by substituting $\delta = \left(\frac{1}{(1+\frac{1}{a-1})^{n-2a+2}}\right)^{1/T}$, and $t_a = T/(n-2a)$ into both sides of the equation. 
    
    The LHS becomes
    \begin{align*}
        \delta^T &= \left(\left(\frac{1}{(1+\frac{1}{a-1})^{n-2a}}\right)^{1/T}\right)^T
        \\ &=\frac{1}{(1+\frac{1}{a-1})^{n-2a+2}}\\ 
        &= (\frac{a-1}{a})^{n-2a+2}.
    \end{align*}

    The RHS becomes:
    \begin{align*}
        \frac{1}{a^2}\frac{(a\delta^{t_a})^{n-2a+2}}{(1+a\delta^{t_a})^{n-2a}} &= \frac{1}{a^2}\frac{a^{n-2a+2}(\frac{a-1}{a})^{n-2a+2}}{(1+a-1)^{n-2a}}\\
        &= \frac{1}{a^2}\frac{(a-1)^{n-2a+2}}{a^{n-2a}}\\
        &= (\frac{a-1}{a})^{n-2a+2}.
    \end{align*}

    As we know that $h_a(t_a) = \frac{1}{a^2}\frac{(a\delta^{t_a})^{n-2a+2}}{(1+a\delta^{t_a})^{n-2a}}$ is monotonically decreasing with $t_a$. Hence, $t_a = T/(n-2a)$ is the only solution to ~\cref{eq: ta in terms of T}.
\end{proof}

\begin{lemma}
    For $\delta \in \left[\left(\frac{1}{(1+\frac{1}{a-1})^{n-2a+2}}\right)^{1/T}, \left(\frac{1}{(1+\frac{1}{a})^{n-2a}}\right)^{1/T}\right]$, the value of $t_a$ which satisfies ~\cref{eq: ta in terms of T} strictly monotonically decreases  with $\delta$ from $\frac{T}{n-2a+2}$ to $0$.
\end{lemma}

\begin{proof}
    We know that when ~\cref{eq: ta in terms of T} holds, we have, 
    \[
    \delta^T = \frac{1}{a^2}\frac{(a\delta^{t_a})^{n-2a+2}}{(1+a\delta^{t_a})^{n-2a}}
    \]
    Differentiating with respect to $\delta$, we obtain that 
    \begin{align*}
        \frac{T}{a^2\delta}\frac{(a\delta^{t-a})^{n-2a+2}}{(1+a\delta^{t_a})^{n-2a}}&= \frac{1}{a^2}\dfrac{(1+a\delta^{t_a})^{n-2a}(n-2a+2)(a\delta^{t_a})^{n-2a+1}a(\delta^{t_a}\ln{\delta}\frac{\partial t_a}{\partial \delta}+t_a\delta^{t_a-1})}{(1+a\delta^{t_a})^{2n-4a}}\\
        &-\frac{1}{a^2}\dfrac{(a\delta^{t_a})^{n-2a+2}(n-2a)(1+a\delta^{t_a})^{n-2a-1}a(\delta^{t_a}\ln{\delta}\frac{\partial t_a}{\partial \delta}+t_a\delta^{t_a-1})}{(1+a\delta^{t_a})^{2n-4a}}
    \end{align*}
    Hence,
    \begin{align*}
        \frac{T}{\delta}(1+a\delta^{t_a})(a\delta^{t_a})&=[(1+a\delta^{t_a})(n-2a+2)a-(a\delta^{t_a})(n-2a)a]\delta^{t_a-1}(\delta\ln{\delta}\frac{\partial t_a}{\partial \delta}+t_1)
    \end{align*}

    Rearranging the terms, we get 
    \begin{align*}
        T(1+a\delta^{t_a})&=[(n-2a)+2(1+a\delta^{t_a})][\delta\ln{\delta}\frac{\partial t_a}{\partial \delta}+t_a]
    \end{align*}

    Therefore, separating variables, we get
    \begin{align*}
        \delta\ln{\delta}\frac{\partial t_a}{\partial\delta}&=[T-(n-2a+2)t_a]+(a\delta^{t_a})[T-2t_a]
    \end{align*}
    Finally, we can write $\partial t_a/ \partial\delta$ as:
    \begin{align*}
        \frac{\partial t_a}{\partial\delta} &= \frac{[T-(n-2a+2)t_a]+(a\delta^{t_a})[T-2t_a]}{\delta\ln{\delta}}
    \end{align*}

    When $\delta = \left(\frac{1}{(1+\frac{1}{a-1})^{n-2a+2}}\right)^{1/T}$, we have for the optimal $\tee$, by ~\cref{lem: val t_a at boundary}, we have that $t_a = \frac{T}{n-2a+2}$. Observe that at this point, the derivative of $T$ with respect to $\delta$ is negative, and it remains negative as long as $t_a \leq T/(n-2a+2)$. Hence, in the domain $\delta \in \left[\left(\frac{1}{(1+\frac{1}{a-1})^{n-2a+2}}\right)^{1/T}, \left(\frac{1}{(1+\frac{1}{a})^{n-2a}}\right)^{1/T}\right]$, we have that $t_a$ decreases strictly monotonically with respect to $\delta$.
\end{proof}




We now prove a lemma similar to  \cref{cor:equalspace}.
\begin{lemma}
    At the minima, for the first non-zero and last non-$T$ time we have that the ads are equispaced, i.e., $t_{a}-t_{a-1} = t_{a}-0 = T-t_{n-a} = t_{n-a+1}-t_{n-a}$.
\end{lemma}
\begin{proof}
    From ~\cref{eq:general t_n}, we obtain that \[
    T_{n-a+1}=\delta^T=\frac{1}{a^2}\frac{(aT_a)^{n-2a+2}}{(1+aT_a)^{n-2a}}
    \]
    From ~\cref{eq:gen t_i}, we obtain that \[
         T_{a+i} = \frac{1}{a}\cdot \frac{(aT_a)^{i+1}}{(aT_a+1)^{i}}
    \]
    Substituting $i=n-2a$, we get
    \[
        T_{n-a} = \frac{1}{a}\cdot \frac{(aT_a)^{n-2a+1}}{(aT_a+1)^{n-2a}}
    \]

    Taking the ratio, we obtain 
    \[
    \frac{T_{n-a+1}}{T_{n-a}} = T_a = \frac{T_{a}}{T_{a-1}}
    \]

   This implies
    \[
    \delta^{t_{n-a+1}-t_{n-a}} = \delta^{t_a-t_{a-1}}
    \]

    Using the fact that $t_{n-a+1} = T$ and $t_{a-1} = 0$, we have
    \[
    \delta^{T-t_{n-a}} = \delta^{t_a-0}
    \]
    Taking $\log$ on both sides, we get
    \[
    T- t_{n-a} = t_a - 0
    \]
    This proves our lemma.
\end{proof}

Now, similar to \cref{lem:equalratio}, we have the following lemma.
\begin{lemma}
    \begin{equation*}
        \frac{T_i}{T_{i-1}} = \frac{aT_1}{aT_1+1}
    \end{equation*}

\end{lemma}
\begin{proof}
From ~\cref{eq:gen t_i}, we obtain that \[
         T_{a+i} = \frac{1}{a}\cdot \frac{(aT_a)^{i+1}}{(aT_a+1)^{i}}
    \]    

    \begin{align*}
        \frac{T_{a+i+1}}{T_{a+i}} &= \frac{1}{a}\cdot \frac{(aT_a)^{i+2}}{(aT_a+1)^{i+1}} \times a \cdot \frac{(aT_a+1)^i}{(aT_a)^{i+1}}\\
        &=\frac{aT_a}{aT_a+1}
    \end{align*}
    
\end{proof}
Hence, the ratio between successive terms is some constant. Now we prove a lemma similar to  \cref{lem:incdec}.
\begin{lemma}
    The value of $t_a$ decreases as $\delta$ increases.
\end{lemma}
\begin{proof}
    From ~\cref{eq:general t_n}, we obtain that \[
    T_{n-a+1}=\frac{1}{a^2}\frac{(aT_a)^{n-2a+2}}{(1+aT_a)^{n-2a}}
    \]

    We have already established that the RHS is a monotonically increasing function in ~\cref{lem: gen montonic increasing}. Observe that $T_{n-a+1} = \delta^T$, therefore as $\delta$ increases, $T_{n-a+1}$ increases. This means that the LHS increases with delta. Since the RHS is monotonic, $T_a$ also increases with $\delta$. As $T_a = \delta^{t_a}$, $t_a$ decreases as $\delta$ increases.
\end{proof}

We now summarize here the behavior of our near-optimal solution.

\begin{observation}\label{obs:behavior1}
The near-optimal ad schedule exhibits the following patterns as $\delta$ varies:
\begin{itemize}
\item As $\delta \to 0$, ads are placed at evenly spaced intervals.
\item As $\delta$ increases, more ads concentrate at times $0$ and $T$.
\item As $\delta$ increases, the first $t_i > 0$ moves towards 0, and the last $t_j < T$ moves towards $T$. The remaining ads are evenly spaced between $t_i$ and $t_j$.
\end{itemize}
\end{observation}

The above properties indicate a form of clustering behavior in our near-optimal solution. When $\delta \to 0$, it is optimal to display the advertisements at uniformly spaced intervals. As $\delta$ increases, a greater number of advertisements are placed at the endpoints, $0$ and $T$, while the remaining ones are evenly distributed in the interior of the interval. In the limit as $\delta \to 1$, the majority of the ads are concentrated at $0$ and $T$, with only a few ads placed uniformly in between.

So far, our analysis has assumed that ads are instantaneous. We now extend our approach to the setting where all ads have equal size and show that a near-optimal schedule can still be achieved.

\subsection{Advertisements with equal sizes}
Let us dive into the case where the ads are of equal size, say $s$. We want that at any point in time, there is only one ad schedule. For this problem, we consider a new loss function which is defined as
\begin{equation}
    L_s(\tee) = \sum_{j < i}\delta^{t_i - t_j - s}
\end{equation}

The above loss function can be written as
\begin{equation}\label{eq:size}
    L_s(\tee) = \delta^{-s}\sum_{j < i}\delta^{t_i - t_j}
\end{equation}

Note that the loss function of \cref{eq:size} is similar to $L(\tee)$. Hence, the minima obtained for $L(\tee)$ will also be minima for $L_s(\tee)$. However, for feasibility, we need that the ads are at least $s$ time apart. Observe that for \cref{eq:size}, we need an interior point solution, and hence we work under the assumption that $T > n\log_{1/\delta}(2)$ (i.e., $a=1$).

\begin{lemma}
    When all ads are equal in length, if the second ad can be scheduled as per our near-optimal schedule, without overlapping with the first ad, then all ads can be scheduled as per our near-optimal schedule.
\end{lemma}
\begin{proof}
    To construct the schedule, we place our first ad at $t_0=0$ and the last ad at $t_n= T-s$. We then solve the problem for $T' = T-s$.

    We now want to show that if the second ad can be scheduled optimally, then all ads can be scheduled optimally. To schedule the second ad, we need that $t_1-t_0 \geq s$, which implies $t_1 \geq s$ and hence, $T_1 = \delta^{t_1} \leq \delta^s$.  This immediately ensures that the second last ad i.e., $t_{n-1}$ can be scheduled as $t_1-t_0=t_n-t_{n-1}$ (using  \cref{lem:firstlastgap}).

We also know that for $i\in \{2, \dots, n-1\}$ 
\begin{align*}
    \frac{T_i}{T_{i-1}} &= \frac{T_1}{1+T_1}\\& \leq T_1 \hspace{1cm}\text{(Using \cref{lem:equalratio})}\\
\end{align*}

Therefore, 
\begin{align*}
    \delta^{t_i-t_{i-1}} &\leq T_1
    \\&\leq \delta^s 
\end{align*}

Taking $\log$ to the base $\delta$ on both sides, we obtain
\[
    t_i - t_{i+1}\geq s
\]

Hence, there is enough space for ads to be scheduled optimally.
\end{proof}

So far, we have assumed that the number of ads $n$ is known. In the next section, we show how to determine the optimal number of ads for a user.

\subsection{Optimizing the number of ads to display}

Given that the relative strengths of mere exposure and operant conditioning are known, a natural question arises: how many ads should be shown to maximize the overall reward? To answer this, we observe that for any fixed number of ads $n$, the loss due to operant conditioning can be computed exactly. In addition, $\sum B(i)$ can be calculated since it is just a function of the number of pulls, allowing us to calculate the total reward $R(\tee)$ precisely.

To find the optimal number of ads, we evaluate $R(\tee)$ for each $n \in \{1, \dots, \tilde{n}\}$, where $\tilde{n}$ is the maximum number of ads that can fit within the time horizon $T$. The value of $n$ that maximizes $R(\tee)$ gives the optimal number of ads. ~\cref{alg: optimal number of ads and schedule} performs this procedure and returns both the optimal number of ads and a corresponding near-optimal schedule in quasi-quadratic time with respect to $\tilde{n}$. Note that in most real-world scenarios, the maximum number of ads is either known in advance or can be learned efficiently. Moreover, prior work \cite{redcircle2023ads,pandorameasuring} suggests that the maximum number of ads is typically not very large. For example, platforms like YouTube and Spotify generally display around 10 and 5 ads per hour, respectively. Hence, we can conclude that ~\cref{alg: optimal number of ads and schedule} performs in quasi-linear time for all practical purposes.

\section{Experiments}\label{sec:experiment}
In this section, we present our experimental results. All experiments were conducted on an 11th-generation Intel Core i5 laptop with 8GB of RAM, and each experiment completed within a few seconds. Since our experiments are designed to be lightweight, they can be easily run on any personal computer. The complete code is available at \href{https://anonymous.4open.science/r/Ads-that-Stick-5E13/README.md}{https://anonymous.4open.science/r/Ads-that-Stick-5E13/README.md}.

 We evaluate the performance of our near-optimal strategy through four distinct experiments:
\begin{enumerate}
    \item How does the near-optimal strategy vary as the value of $\delta$ changes?

    \item Comparison between our strategy and other baseline strategies?

    \item How does the loss function change with a change in the number of ads? 

    \item How to find the optimal number of ads?
\end{enumerate}

\subsection{Variation in near optimal strategy with change in $\delta$}\label{sec:ex_first}
In this experiment, we illustrate how our strategy evolves as the parameter $\delta$ increases from 0 to 1. ~\cref{fig:exp1a} depicts the outcome when the number of ads $n$ is odd.
Initially, for $\delta \leq 0.4$, the ads are placed nearly equidistantly. As $\delta$ increases beyond 0.4, the ads gradually bifurcate—half of them shift towards $t = 0$, and the other half move towards $t = 20$. This experiment aligns precisely with the behavior we obtained in ~\cref{obs:behavior1}.

When $n$ is even, we observe that the pattern is largely similar with the odd case, with one key difference: for odd number of ads, the $\frac{n+1}{2}$-th ad remains fixed at $t = T/2$, regardless of the value of $\delta$ (See ~\cref{fig:exp1a}). But for an even number of ads (See  \cref{fig:exp1b}), the optimal times to show the ads in the first half shift towards $0$ while the optimal time to show ads in the second half shifts towards $T$ as the value of $\delta$ increases. Hence, it also supports the theoretical behavior described in ~\cref{obs:behavior1}.

\begin{figure}
    \centering
    \begin{subfigure}{0.48\textwidth}
        \includegraphics[width=\textwidth]{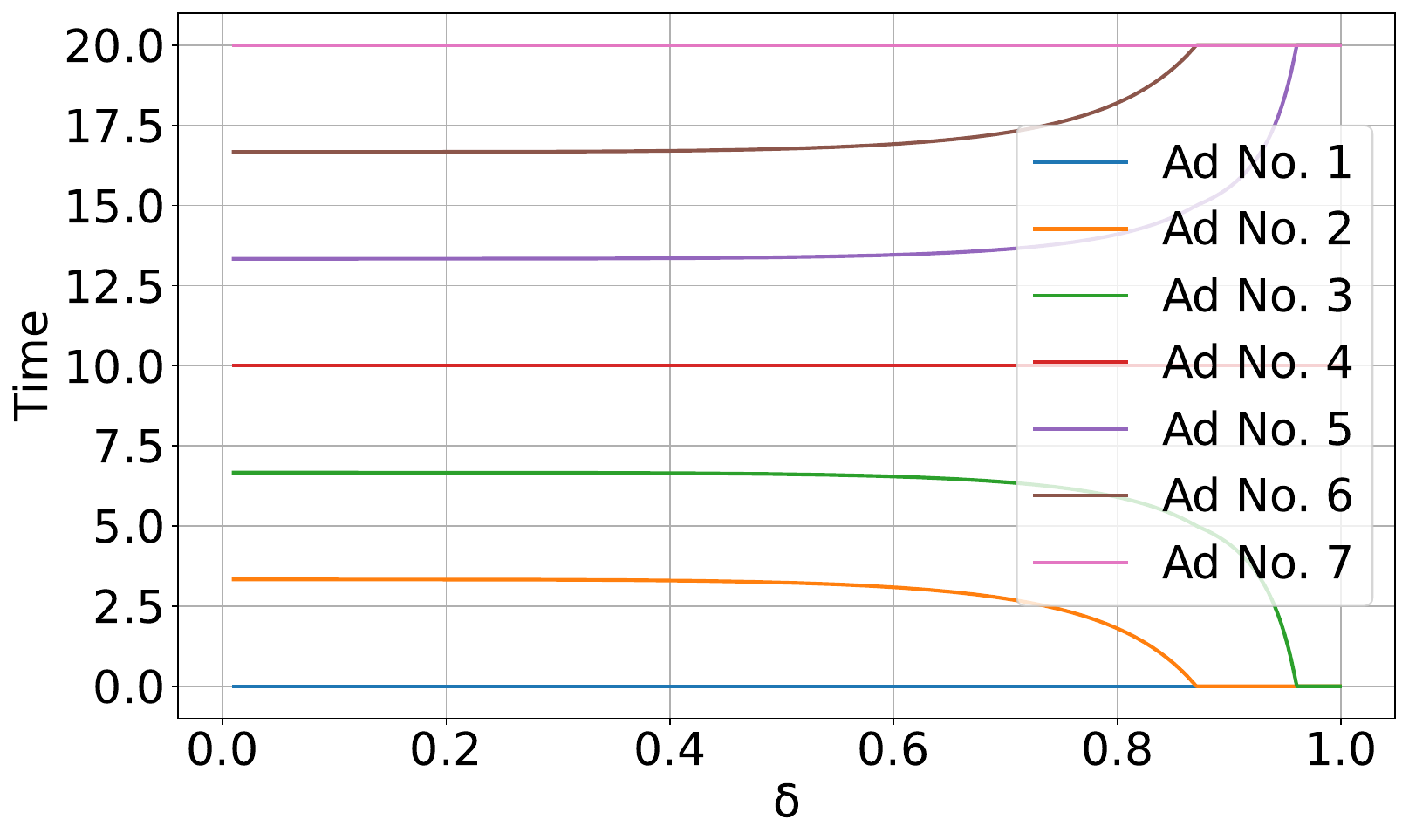}
        \caption{Our strategy for $n=7$, $T=20$}
        \label{fig:exp1a}
    \end{subfigure}
   \begin{subfigure}{0.48\textwidth}
    \centering
    \includegraphics[width=\textwidth]{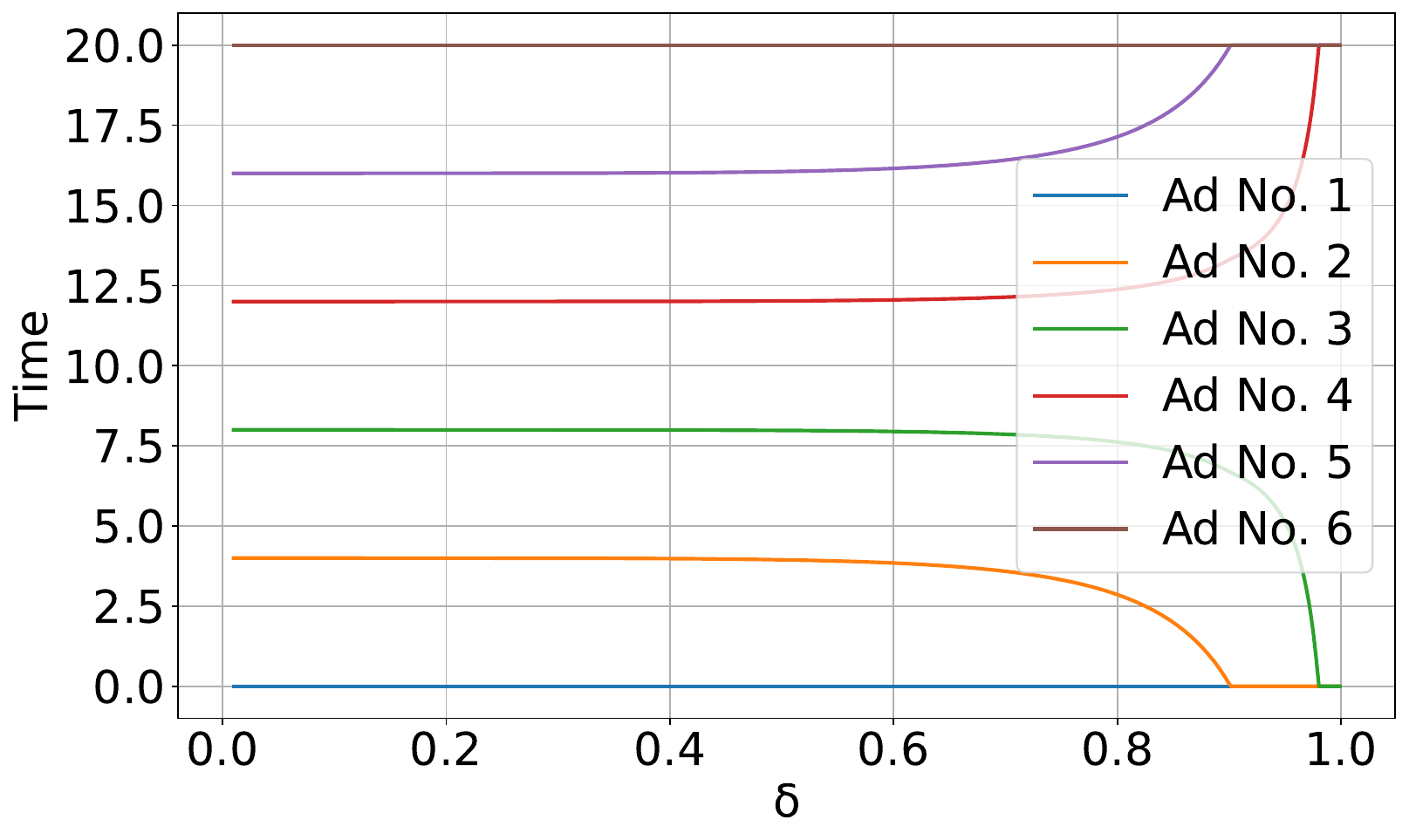}
    \caption{Our strategy for $n=6$, $T=20$}
    \label{fig:exp1b}
\end{subfigure}
    \caption{
        Change in near-optimal strategy with $\delta$ for odd and even numbers of ads.
    }
    \label{fig:exp1}
\end{figure}

\subsection{Our strategy vs baseline strategies}\label{sec:ex_second}

We have presented the performance of our near-optimal strategy compared to three baseline strategies, namely Uniform, Corner, and Random, as described below.
\begin{itemize}
\item \textbf{Uniform:} Ads are placed uniformly over $[0, T]$, with equal spacing.

\item \textbf{Corner:} Half of the ads are placed at $t = 0$, half at $t = T$. 
\item \textbf{Random:} The first and last ads are fixed at $t = 0$ and $t = T$, respectively, while the remaining $n-2$ ads are distributed uniformly at random over the interval $(0, T)$.
\end{itemize}

In \cite{curmei2022towards}, the authors note that in real-world scenarios, the value of $\delta$ is usually high, around 0.98. Other studies \cite{murre2015replication,goldstein2011effects} on ads also suggest that $\delta$ typically lies between 0.7 and 0.99. That said, our theoretical framework is valid for any value of delta and gives a near-optimal solution. We have mainly used high $\delta$ values in this experiment based on what is reported in the literature.

Our first experiment (\cref{fig:exp2a}) models a video streaming setting, where users engage for 1.5–2 hours and are shown approximately 15 ads. Therefore, we have $n=15$ and $T=100$. For this experiment, we focus on a high value of $\delta$ (greater than $0.9$). As observed in the previous experiment (\cref{fig:exp1}) and in ~\cref{obs:behavior1}, when $\delta$ is small, the Uniform strategy performs well. Conversely, when $\delta$ approaches 1, the Corner strategy becomes more effective. Our experiment (\cref{fig:exp2}) shows that the near-optimal strategy adapts to $\delta$ and consistently outperforms all baseline strategies. Note that even though the difference in reward seems small, it can have a large impact on customer retention and revenue.

 We now present another experiment, but for a music streaming setting (\cref{fig:exp2b}), where a typical session lasts around 60 minutes with approximately 6 ads. This is to emulate what popular music streaming services offer in their free plan. It can be observed that even in this case, our strategy performs better than all the baseline strategies.

\begin{figure}
    \centering
    \begin{subfigure}{0.48\textwidth}
        \includegraphics[width=\textwidth]{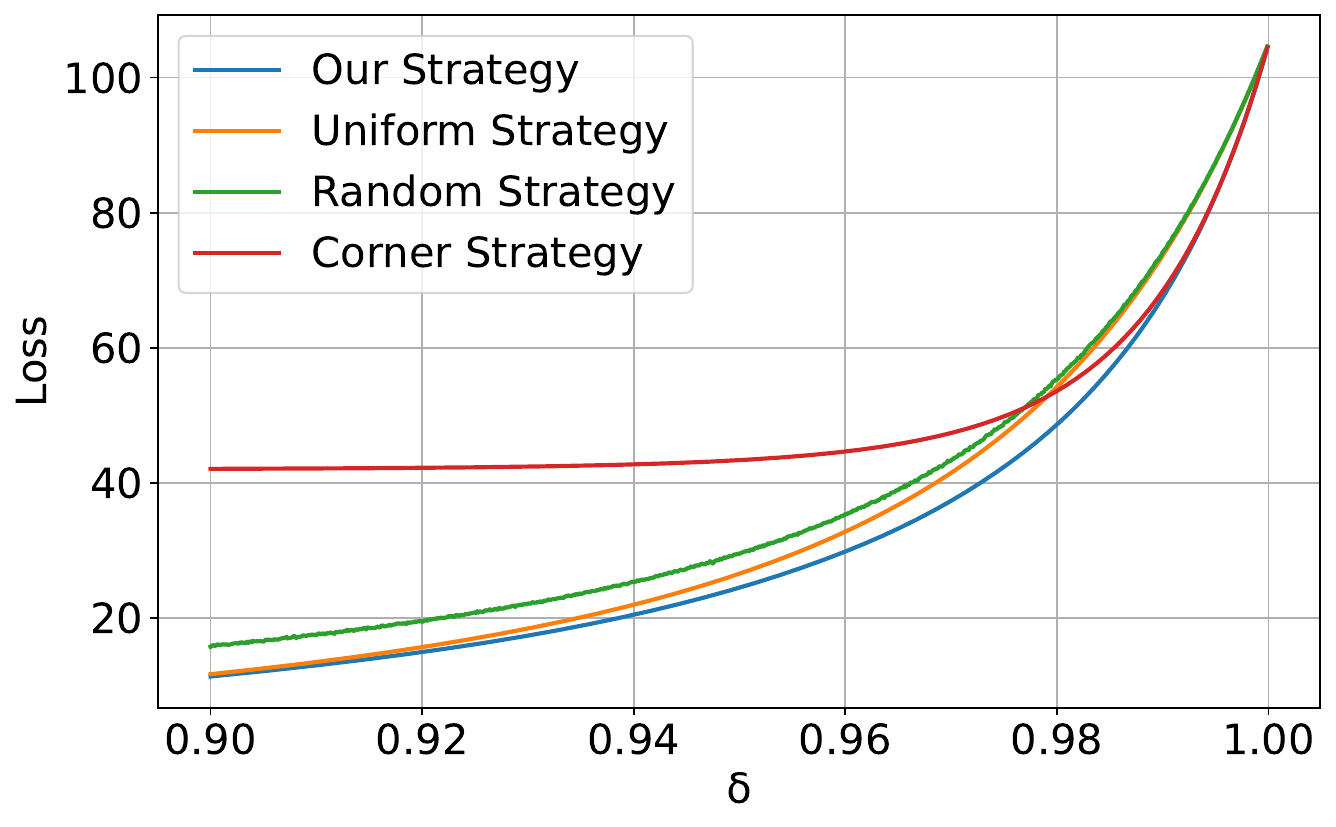}
        \caption{Loss vs $\delta$ for $n=15$, $T=100$}
        \label{fig:exp2a}
    \end{subfigure}
   \begin{subfigure}{0.48\textwidth}
    \centering
    \includegraphics[width=\textwidth]{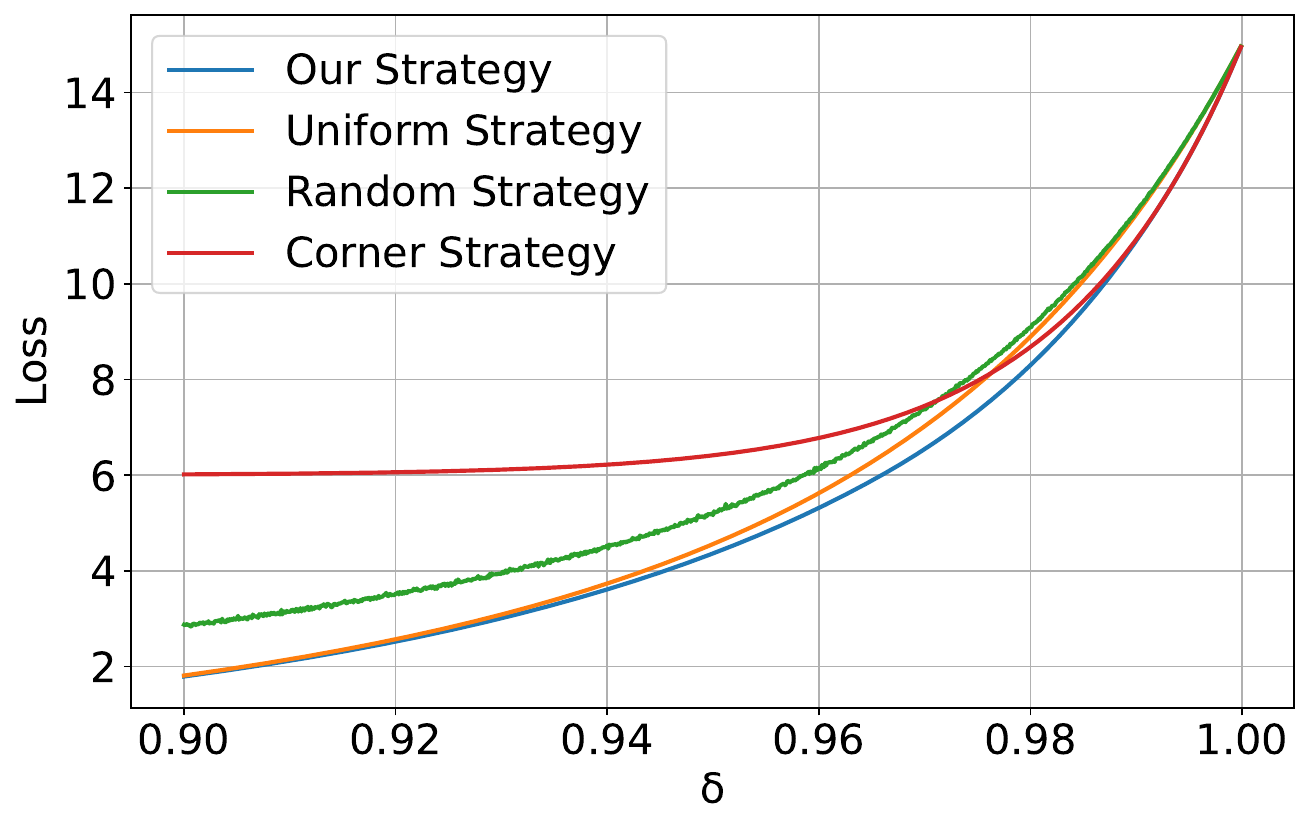}
    \caption{Loss with $\delta$ for $n=6$, $T=60$}
    \label{fig:exp2b}
\end{subfigure}
    \caption{
        Loss between near-optimal and baseline strategies for $n=15$, $T=100$ and $n=6$, $T=60$.
    }
    \label{fig:exp2}
\end{figure}

\subsection{Change in loss with number of ads}\label{sec:ex_third}
We now conduct another experiment to quantify the loss incurred by our near-optimal strategy due to the effect of operant conditioning, as the number of ads increases. Let $L^{\#}(n)$ denote the loss associated with showing $n$ ads under our strategy. A natural intuition is that if showing $n$ ads results in a loss of $L^{\#}(n)$, then doubling the number of ads would roughly double the loss, i.e., $L^{\#}(2n) \approx 2 \cdot L^{\#}(n)$. Our experiments actually support this intuition (see \cref{fig:exp3}).

Interestingly, for smaller values of $\delta$ (around 0.7), the loss remains relatively stable even as the number of ads increases. However, we observe a sharp rise in loss as $\delta$ increases from 0.9 to 0.99, indicating increased sensitivity to operant conditioning in this regime (see  \cref{fig:exp3a}). We further extended our experiment by plotting $\log(\text{loss})$ versus the number of ads (~\cref{fig:exp3b}), and observed trends similar to before.

\begin{figure}
    \centering
    \begin{subfigure}{0.48\textwidth}
        \includegraphics[width=\textwidth]{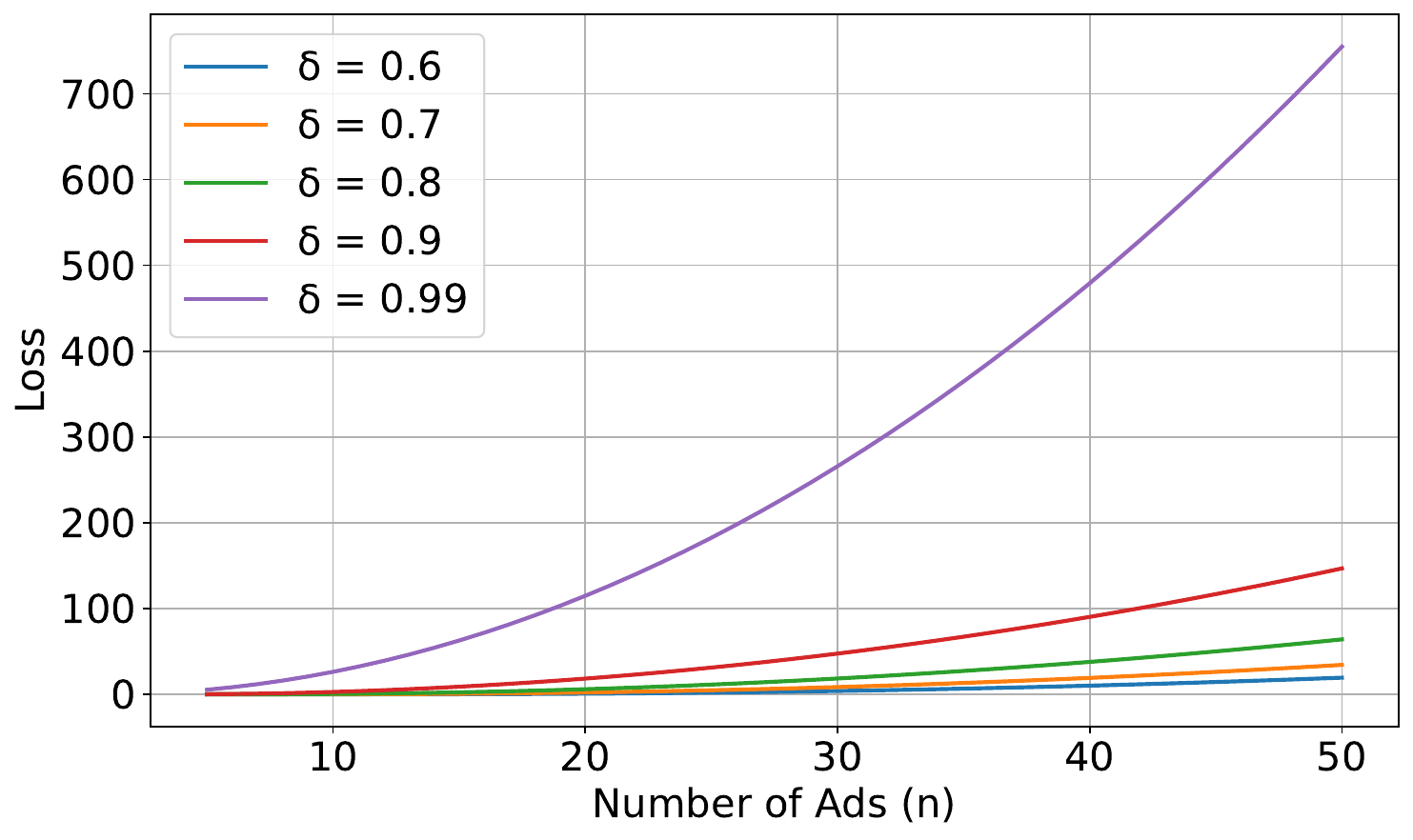}
        \caption{Loss vs $n$\_ads}
        \label{fig:exp3a}
    \end{subfigure}
   \begin{subfigure}{0.48\textwidth}
    \includegraphics[width=\textwidth]{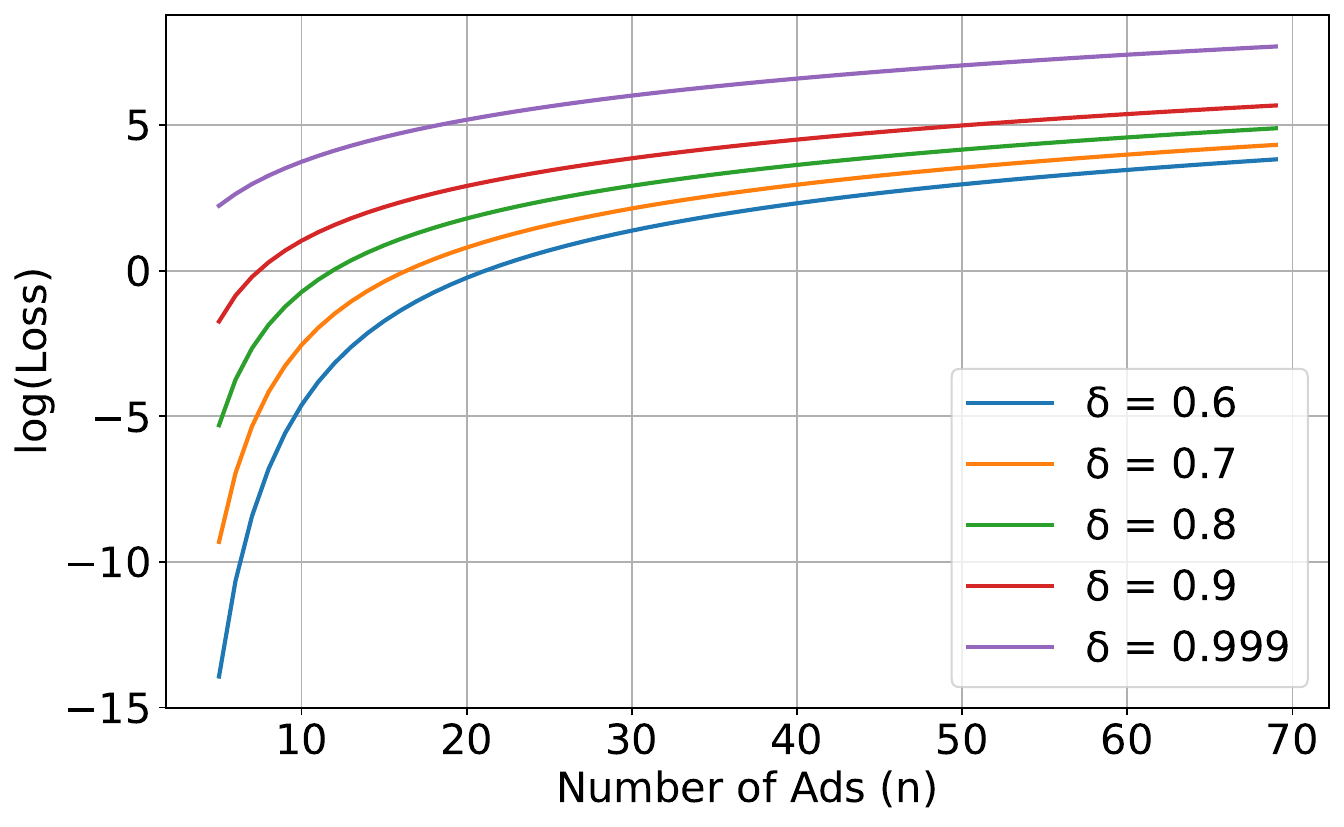}
    \caption{$\log$(loss) vs $n$\_ads}
    \label{fig:exp3b}
\end{subfigure}
    \caption{
       Change in loss and $\log$(loss) for different values of $\delta$ with an increase in the number of ads.
    }
    \label{fig:exp3}
\end{figure}

\subsection{Optimum number of ads}\label{sec:ex_four}

In our final experiment, we demonstrate that for a given user (fixed $\delta$), the optimal number of ads $n$ changes under different mere exposure and hedonic adaptation functions. We first use a sigmoid reward function $B(i) = k \cdot \frac{1}{1 + e^{-ci}}$, where $k$ captures the overall strength of these effects and $c$ controls sensitivity to the number of ads. For $k, c > 0$, the function is concave and increasing.

As shown in \cref{fig:exp4a}, the reward initially rises with more ads due to the dominance of mere exposure. Beyond a point, the negative impact of hedonic adaptation and operant conditioning becomes significant, causing the reward to decline. The peak of this curve corresponds to the optimal number of ads. We perform another experiment using $B(i) = k(1 - e^{-cx})$ (See ~\cref{fig:exp4b}), and show that the peak of the curve corresponds to the optimal number of ads.

\begin{figure}
    \centering
     \begin{subfigure}{0.48\textwidth}
        \includegraphics[width=\textwidth]{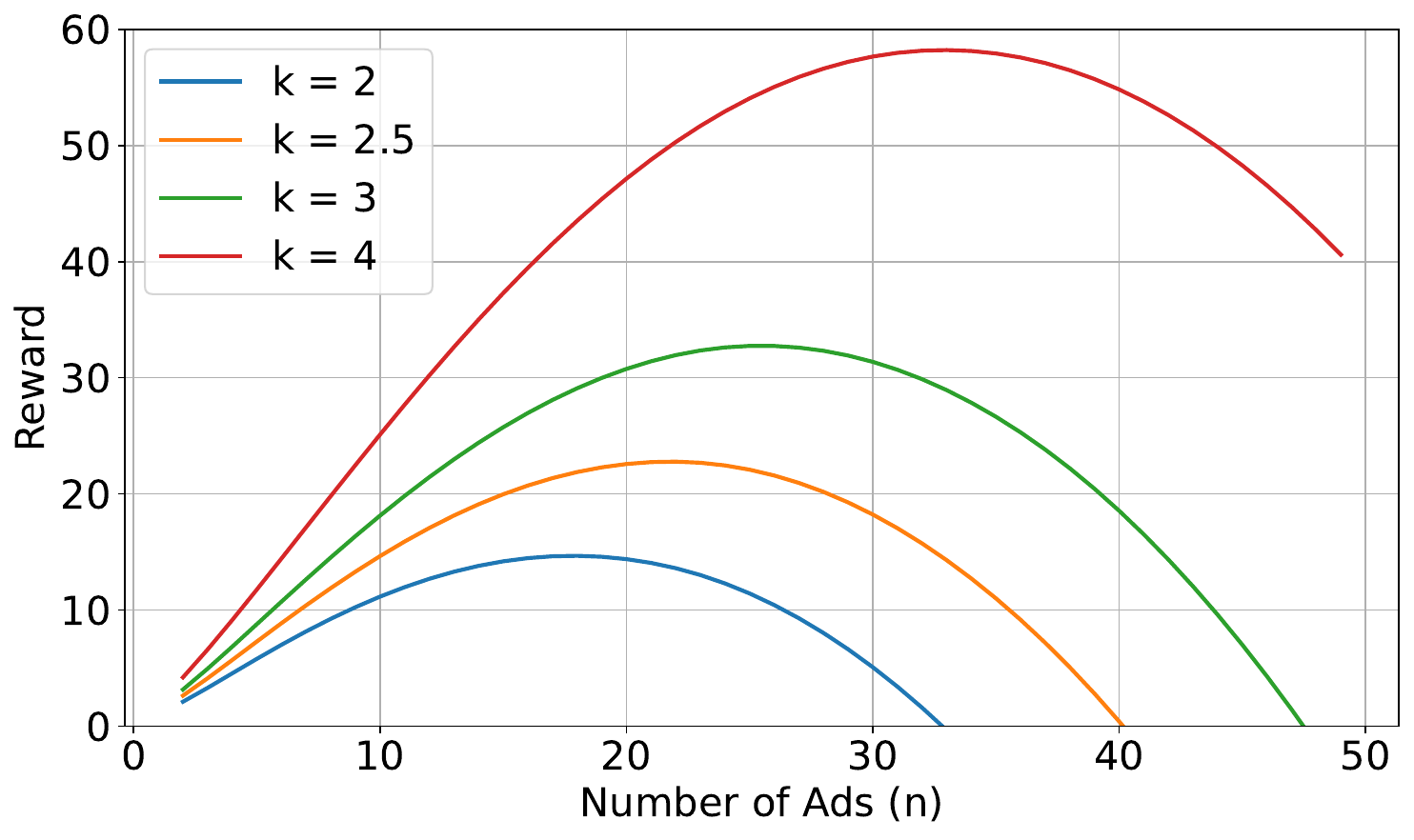}
        \caption{Reward v/s $n$ for $B(i)=k\cdot\frac{1}{1+e^{-i/5}}$}
        \label{fig:exp4a}
    \end{subfigure}
   \begin{subfigure}{0.48\textwidth}
    \includegraphics[width=\textwidth]{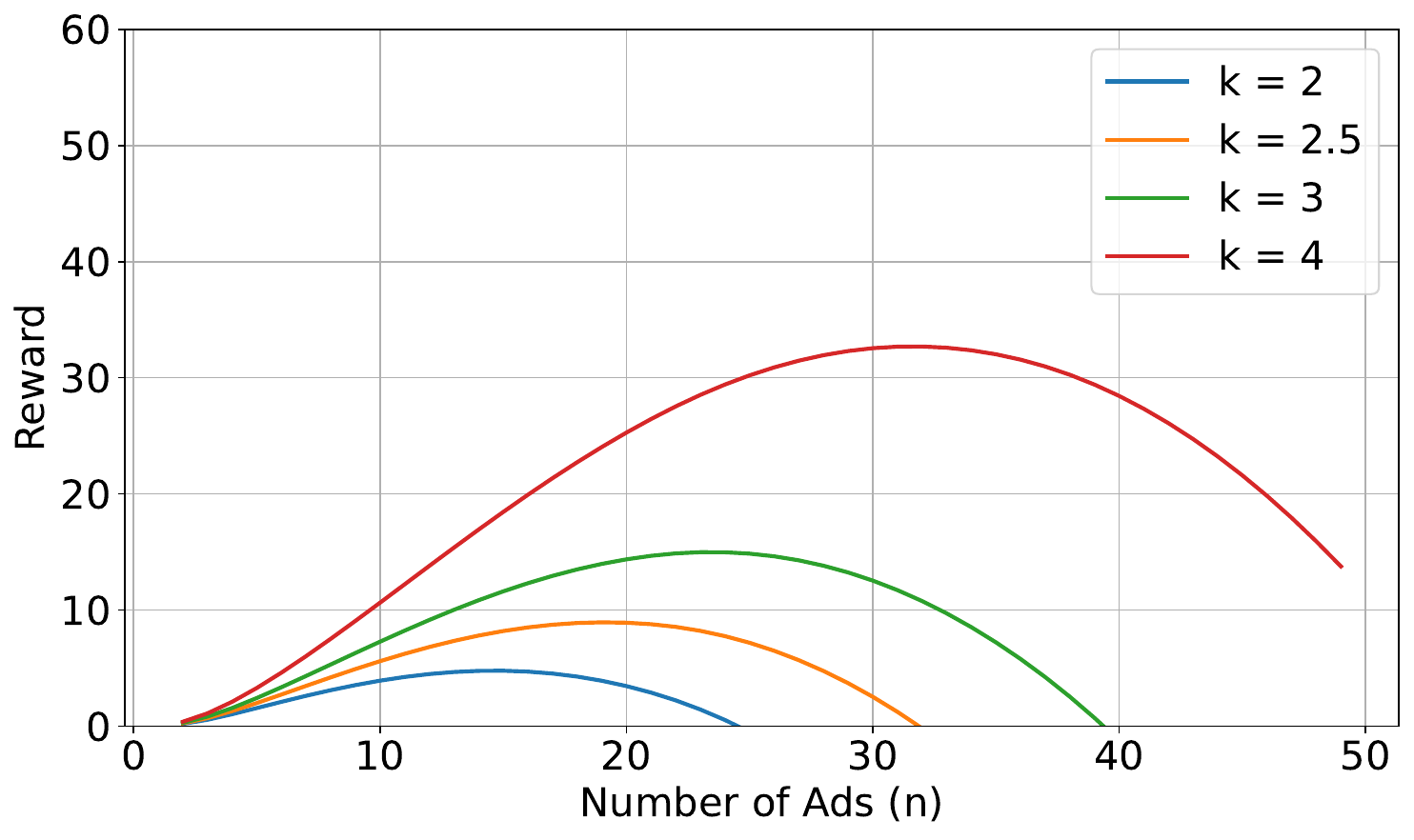}
    \caption{Reward v/s $n$ for $B(i)=k \cdot(1-e^{-i/10})$}
    \label{fig:exp4b}
\end{subfigure}
    \caption{
       Gain functions used to find the optimal number of ads for a user with $\delta = 0.9$.
    }
    \label{fig:exp4}
\end{figure}

\section{Conclusion}\label{sec:conclusion}\label{sec:conclusion}
In this paper, we consider a model that incorporates dynamic psychological effects -- mere exposure, hedonic adaptation, and operant conditioning into the problem of ad scheduling.
We present a near-optimal strategy to schedule ads based on our behavioral model. Our strategy leads to several insights into the problem of ad scheduling, for example, equal spacing of ads might not be optimal under many settings, and it might be better to show more ads in the beginning and at the end as compared to the middle of the time-horizon.
We also support these theoretical results using simulations.  

\subsection{Limitations/Extensions}
\noindent
\textbf{Seasonality and non-stationary rewards.}
While our model can work well for real-world settings where the rewards are approximately stationary, such as inserting ads into a (live) video stream,
our model does not handle scenarios where the rewards are affected by seasonality or time-of-day effect.
For example, it is unlikely that sending a push notification at night will result in user attention.

It will be interesting to extend our model to account for non-stationary rewards.

\noindent
\textbf{Competition between ads.}
In our work, we consider the optimization of the ad schedule from the point of view of a single advertiser or an ad agency running multiple homogeneous ads. 
In the future it will be interesting to account for externalities in the form of competing ads across various channels.

\noindent
\textbf{Incorporating context or side-information.}
Our model does not incorporate the context or side-information of the user or the ad into the optimization problem.
This is motivated by the fact that under many scenarios, advertisers do not have access to user information, such as advertising on streaming platforms.
Another future direction is to incorporate additional context of the user and ad.

\noindent
\textbf{Learning the reward function.}
Our current setup assumes knowledge of the reward function (including the parameter $\delta$). While our methodology is flexible enough to allow various types of reward functions, it will be interesting to study our problem as a
joint learning and optimization problem in a multi-armed bandits setting.

\bibliographystyle{alpha}

\begin{thebibliography}{CGK{\etalchar{+}}15}

\bibitem[CHRHM22]{curmei2022towards}
M. Curmei, A. A. Haupt, B. Recht, and D. Hadfield-Menell.
\newblock Towards psychologically-grounded dynamic preference models.
\newblock In \emph{Proceedings of the 16th ACM Conference on Recommender Systems}, pages 35--48, 2022.

\bibitem[Raf23]{rafieian2023optimizing}
O. Rafieian.
\newblock Optimizing user engagement through adaptive ad sequencing.
\newblock \emph{Marketing Science}, 42(5):910--933, 2023.

\bibitem[NMS98]{naik1998planning}
P. A. Naik, M. K. Mantrala, and A. G. Sawyer.
\newblock Planning media schedules in the presence of dynamic advertising quality.
\newblock \emph{Marketing Science}, 17(3):214--235, 1998.

\bibitem[AN15]{aravindakshan2015understanding}
A. Aravindakshan and P. A. Naik.
\newblock Understanding the memory effects in pulsing advertising.
\newblock \emph{Operations Research}, 63(1):35--47, 2015.

\bibitem[SMBL94]{singh1994enhancing}
S. N. Singh, S. Mishra, N. Bendapudi, and D. Linville.
\newblock Enhancing memory of television commercials through message spacing.
\newblock \emph{Journal of Marketing Research}, 31(3):384--392, 1994.

\bibitem[Sah15]{sahni2015effect}
N. S. Sahni.
\newblock Effect of temporal spacing between advertising exposures: Evidence from online field experiments.
\newblock \emph{Quantitative Marketing and Economics}, 13:203--247, 2015.

\bibitem[NA62]{nerlove1962optimal}
M. Nerlove and K. J. Arrow.
\newblock Optimal advertising policy under dynamic conditions.
\newblock \emph{Economica}, pages 129--142, 1962.

\bibitem[Ebb13]{ebbinghaus1913contribution}
H. Ebbinghaus.
\newblock A contribution to experimental psychology.
\newblock \emph{New York, NY: Teachers College, Columbia University}, 1913.

\bibitem[JJ21]{jesse2021digital}
M. Jesse and D. Jannach.
\newblock Digital nudging with recommender systems: Survey and future directions.
\newblock \emph{Computers in Human Behavior Reports}, 3:100052, 2021.

\bibitem[HTW13]{hekkert2013mere}
P. Hekkert, C. Thurgood, and T. W. A. Whitfield.
\newblock The mere exposure effect for consumer products as a consequence of existing familiarity and controlled exposure.
\newblock \emph{Acta Psychologica}, 144(2):411--417, 2013.

\bibitem[CHH+07]{cooper2007applied}
J. O. Cooper, T. E. Heron, W. L. Heward, et al.
\newblock Applied behavior analysis.
\newblock Pearson/Merrill-Prentice Hall, Upper Saddle River, NJ, 2007.

\bibitem[TTG15]{TheocharousTG15}
G. Theocharous, P. S. Thomas, and M. Ghavamzadeh.
\newblock Personalized Ad Recommendation Systems for Life-Time Value Optimization with Guarantees.
\newblock In Q. Yang and M. J. Wooldridge, editors, \emph{Proceedings of the Twenty-Fourth International Joint Conference on Artificial Intelligence (IJCAI 2015)}, pages 1806--1812. AAAI Press, 2015.

\bibitem[DRS21]{despotakis2021first}
S. Despotakis, R. Ravi, and A. Sayedi.
\newblock First-price auctions in online display advertising.
\newblock \emph{Journal of Marketing Research}, 58(5):888--907, 2021.

\bibitem[BV04]{boyd2004convex}
S. Boyd and L. Vandenberghe.
\newblock \emph{Convex Optimization}.
\newblock Cambridge University Press, 2004.

\bibitem[SBF17]{schwartz2017customer}
E. M. Schwartz, E. T. Bradlow, and P. S. Fader.
\newblock Customer acquisition via display advertising using multi-armed bandit experiments.
\newblock \emph{Marketing Science}, 36(4):500--522, 2017.

\bibitem[AKO13]{adany2013allocation}
R. Adany, S. Kraus, and F. Ordonez.
\newblock Allocation algorithms for personal TV advertisements.
\newblock \emph{Multimedia Systems}, 19:79--93, 2013.

\bibitem[SSS15]{seshadri2015scheduling}
S. Seshadri, S. Subramanian, and S. Souyris.
\newblock Scheduling spots on television.
\newblock 2015. Available at: \url{http://www.ssouyris.com/}.

\bibitem[DOM25]{dobrita2025federated}
Gabriela Dobrița, Simona-Vasilica Oprea, and Adela Bâra.
\newblock Federated learning's role in next-gen TV ad optimization.
\newblock \emph{Computers, Materials \& Continua}, 82(1), 2025.


\bibitem[Cze19]{czerniachowska2019scheduling}
K. Czerniachowska.
\newblock Scheduling TV advertisements via genetic algorithm.
\newblock \emph{European Journal of Industrial Engineering}, 13(1):81--116, 2019.

\bibitem[SKF+17]{sumita2017online}
H. Sumita, Y. Kawase, S. Fujita, and T. Fukunaga.
\newblock Online Optimization of Video-Ad Allocation.
\newblock In \emph{IJCAI}, pages 423--429, 2017.

\bibitem[CC02]{cox2002beyond}
D. Cox and A. D. Cox.
\newblock Beyond first impressions: The effects of repeated exposure on consumer liking of visually complex and simple product designs.
\newblock \emph{Journal of the Academy of Marketing Science}, 30(2):119--130, 2002.

\bibitem[FSA07]{fang2007examination}
X. Fang, S. Singh, and R. Ahluwalia.
\newblock An examination of different explanations for the mere exposure effect.
\newblock \emph{Journal of Consumer Research}, 34(1):97--103, 2007.

\bibitem[SH23]{sguerra2023ex2vec}
B. Sguerra, V.-A. Tran, and R. Hennequin.
\newblock Ex2vec: Characterizing users and items from the mere exposure effect.
\newblock In \emph{Proceedings of the 17th ACM Conference on Recommender Systems}, pages 971--977, 2023.

\bibitem[NM08]{nelson2008interrupted}
L. D. Nelson and T. Meyvis.
\newblock Interrupted consumption: Disrupting adaptation to hedonic experiences.
\newblock \emph{Journal of Marketing Research}, 45(6):654--664, 2008.

\bibitem[FAF11]{fagerstrom2011study}
A. Fagerstr{\o}m, E. Arntzen, and G. R. Foxall.
\newblock A study of preferences in a simulated online shopping experiment.
\newblock \emph{The Service Industries Journal}, 31(15):2603--2615, 2011.

\bibitem[Fox17]{foxall2017behavioral}
G. R. Foxall.
\newblock Behavioral economics in consumer behavior analysis.
\newblock \emph{The Behavior Analyst}, 40:309--313, 2017.

\bibitem[CIR15]{chugani2015happily}
S. K. Chugani, J. R. Irwin, and J. P. Redden.
\newblock Happily ever after: The effect of identity-consistency on product satiation.
\newblock \emph{Journal of Consumer Research}, 42(4):564--577, 2015.

\bibitem[YG15]{yang2015sentimental}
Y. Yang and J. Galak.
\newblock Sentimental value and its influence on hedonic adaptation.
\newblock \emph{Journal of Personality and Social Psychology}, 109(5):767, 2015.

\bibitem[LKKLM21]{leqi2021rebounding}
L. Leqi, F. K. Karzan, Z. Lipton, and A. Montgomery.
\newblock Rebounding bandits for modeling satiation effects.
\newblock \emph{Advances in Neural Information Processing Systems}, 34:4003--4014, 2021.

\bibitem[HKR16]{heidari2016tight}
H. Heidari, M. J. Kearns, and A. Roth.
\newblock Tight Policy Regret Bounds for Improving and Decaying Bandits.
\newblock In \emph{IJCAI}, pages 1562--1570, 2016.

\bibitem[LCM17]{levine2017rotting}
N. Levine, K. Crammer, and S. Mannor.
\newblock Rotting bandits.
\newblock \emph{Advances in Neural Information Processing Systems}, 30, 2017.

\bibitem[SLC+19]{seznec2019rotting}
J. Seznec, A. Locatelli, A. Carpentier, A. Lazaric, and M. Valko.
\newblock Rotting bandits are no harder than stochastic ones.
\newblock In \emph{AISTATS}, pages 2564--2572. PMLR, 2019.

\bibitem[KI18]{kleinberg2018recharging}
R. Kleinberg and N. Immorlica.
\newblock Recharging bandits.
\newblock In \emph{2018 IEEE 59th Annual Symposium on Foundations of Computer Science (FOCS)}, pages 309--319. IEEE, 2018.

\bibitem[CCB20]{cella2020stochastic}
L. Cella and N. Cesa-Bianchi.
\newblock Stochastic bandits with delay-dependent payoffs.
\newblock In \emph{AISTATS}, pages 1168--1177. PMLR, 2020.

\bibitem[BSSS19]{basu2019blocking}
S. Basu, R. Sen, S. Sanghavi, and S. Shakkottai.
\newblock Blocking bandits.
\newblock \emph{Advances in Neural Information Processing Systems}, 32, 2019.

\bibitem[WLM18]{warlop2018fighting}
R. Warlop, A. Lazaric, and J. Mary.
\newblock Fighting boredom in recommender systems with linear reinforcement learning.
\newblock \emph{Advances in Neural Information Processing Systems}, 31, 2018.

\bibitem[MAK+20]{mintz2020nonstationary}
Y. Mintz, A. Aswani, P. Kaminsky, E. Flowers, and Y. Fukuoka.
\newblock Nonstationary bandits with habituation and recovery dynamics.
\newblock \emph{Operations Research}, 68(5):1493--1516, 2020.

\bibitem[PNGK23]{patil2022mitigating}
V. Patil, V. Nair, G. Ghalme, and A. Khan.
\newblock Mitigating disparity while maximizing reward: tight anytime guarantee for improving bandits.
\newblock In \emph{Proceedings of the 32nd International Joint Conference on Artificial Intelligence}, pages 4100--4108, 2023.

\bibitem[BR25]{blumnearly}
A. Blum and K. Ravichandran.
\newblock Nearly-tight Approximation Guarantees for the Improving Multi-Armed Bandits Problem.
\newblock In \emph{Algorithmic Learning Theory}, pages 228--245. PMLR, 2025.

\bibitem[FWYS22]{freeman2022does}
J. Freeman, L. Wei, H. Yang, and F. Shen.
\newblock Does in-stream video advertising work? Effects of position and congruence on consumer responses.
\newblock \emph{Journal of Promotion Management}, 28(5):515--536, 2022.

\bibitem[RC09]{ritter2009effects}
E. A. Ritter and C.-H. Cho.
\newblock Effects of ad placement and type on consumer responses to podcast ads.
\newblock \emph{CyberPsychology \& Behavior}, 12(5):533--537, 2009.

\bibitem[GMS11]{goldstein2011effects}
D. G. Goldstein, R. P. McAfee, and S. Suri.
\newblock The effects of exposure time on memory of display advertisements.
\newblock In \emph{Proceedings of the 12th ACM Conference on Electronic Commerce}, pages 49--58, 2011.

\bibitem[MD15]{murre2015replication}
J. M. J. Murre and J. Dros.
\newblock Replication and analysis of Ebbinghaus’ forgetting curve.
\newblock \emph{PLoS One}, 10(7):e0120644, 2015.

\bibitem[Pan]{pandorameasuring}
Sirius XM Pandora.
\newblock Measuring Consumer Sensitivity to Audio Advertising: A Long-Run Field Experiment on Pandora Internet Radio.

\bibitem[Red23]{redcircle2023ads}
RedCircle.
\newblock Less is More: How Many Ads Should There Be in an Ad Break?
\newblock 2023. Available at: \url{https://redcircle.com/blog/less-is-more-how-many-ads-should-there-be-in-an-ad-break/}. Accessed: 2025-09-02.

\end{thebibliography}
\newcommand{\etalchar}[1]{$^{#1}$}

\end{document}